\documentclass[runningheads]{llncs}
\usepackage{graphicx}
\usepackage{subfigure}
\usepackage{comment}
\usepackage{amsmath,amssymb} 
\usepackage{color}
\usepackage{mathrsfs}
\usepackage{multirow}
\usepackage{booktabs}
\usepackage{appendix}
\usepackage{bm}
\usepackage{algpseudocode}
\usepackage{algorithmicx,algorithm}
\long\def\comment#1{}

\usepackage{hyperref}

\begin{document}
\pagestyle{headings}
\mainmatter
\def\ECCVSubNumber{2463}  

\title{Targeted Attack for Deep Hashing based Retrieval} 


\titlerunning{Targeted Attack for Deep Hashing based Retrieval} 
%

\author{
Jiawang Bai \inst{1,2,\ast} \and
Bin Chen \inst{1,2,\ast}  \and
Yiming Li \inst{1,\ast}   \and
Dongxian Wu \inst{1,2} \and \\
Weiwei Guo \inst{3} \and
Shu-Tao Xia \inst{1,2} \and
En-Hui Yang \inst{4}}

\renewcommand{\thefootnote}{\fnsymbol{footnote}}
\footnotetext{$^\ast$Equal contribution.  \\
Correspondence to: Bin Chen (\href{mailto:cb17@mails.tsinghua.edu.cn} {cb17@mails.tsinghua.edu.cn}).}

\authorrunning{J. Bai, B. Chen and Y. Li et al.}

\institute{Tsinghua Shenzhen International Graduate School, Tsinghua University \and 
PCL Research Center of Networks and Communications, Peng Cheng Laboratory \and
vivo AI Lab \and
Department of Electrical and Computer Engineering, University of Waterloo
}

\maketitle

\begin{abstract}
The deep hashing based retrieval method is widely adopted in large-scale image and video retrieval. However, there is little investigation on its security.
In this paper, we propose a novel method, dubbed deep hashing targeted attack (DHTA), to study the targeted attack on such retrieval. 
Specifically, we first formulate the targeted attack as a \emph{point-to-set} optimization, which minimizes the average distance between the hash code of an adversarial example and those of a set of objects with the target label. Then we design a novel \emph{component-voting scheme} to obtain an \emph{anchor code}
as the representative of the set of hash codes of objects with the target label, whose optimality guarantee is also theoretically derived. 
To balance the performance and perceptibility, we propose to minimize the Hamming distance between the hash code of the adversarial example and the anchor code under the $\ell^\infty$ restriction on the perturbation. 
Extensive experiments verify that DHTA is effective in attacking both deep hashing based image retrieval and video retrieval. 

\end{abstract}

\section{Introduction}

High-dimension and large-scale data approximate nearest neighbor (ANN) retrieval has been widely adopted in online search engines, $e.g.$, Google or Bing, due to its efficiency and effectiveness. Within all ANN retrieval methods, hashing-based methods \cite{wang2017survey} have attracted a lot of attentions due to their compact binary representations and rapid similarity computation between hash codes with Hamming distance. In particular, deep learning based hashing methods \cite{liu2016deep,cao2018deep,chen2018deep,hu2018deep,shen2018unsupervised,li2018self} have shown their superiority in performance since they generally learn more meaningful semantic hash codes through learnable hashing functions with deep neural networks (DNNs).

\begin{figure}[ht]
	\centering
	\includegraphics[width=0.8\textwidth]{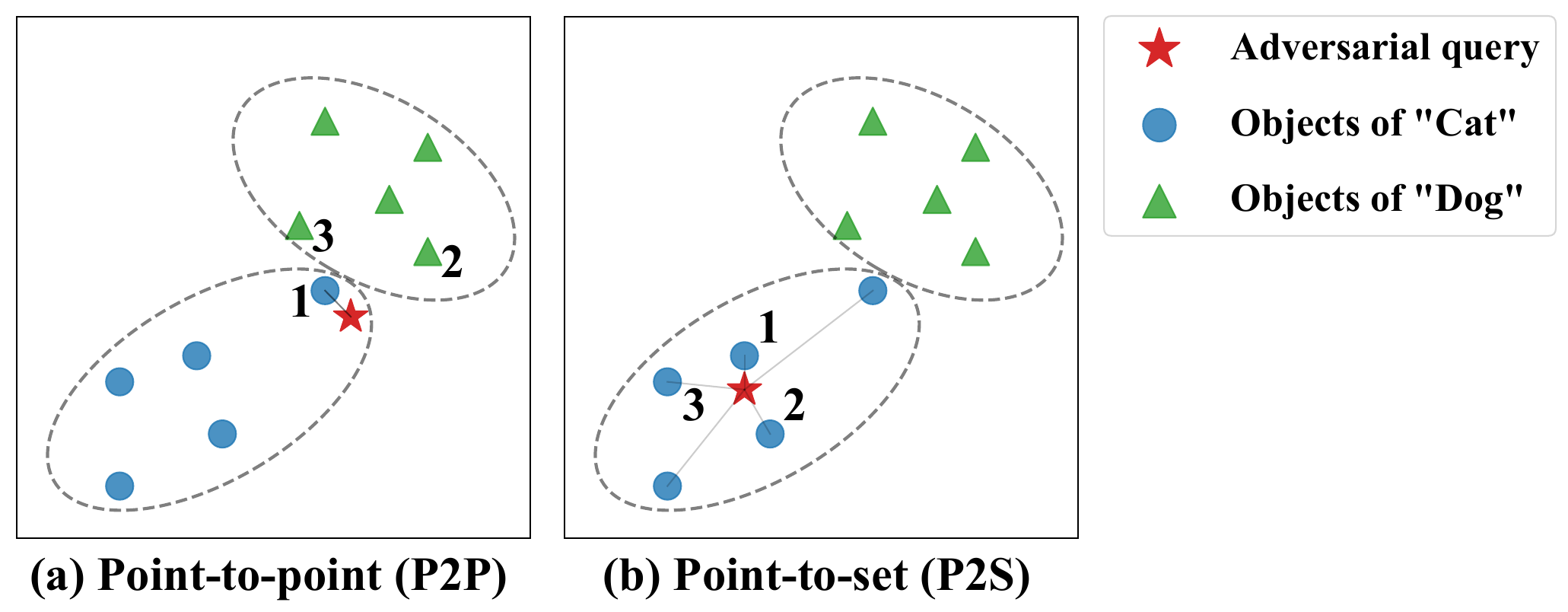}
	\caption{The comparison between the P2P attack paradigm and proposed P2S paradigm. 
	There are two object classes ($i.e.$ `Cat' and `Dog') as shown above, where the target label being attack is `Cat'. In the P2P paradigm, a object with the target label is randomly selected as the reference to generate the adversarial query. But if the selected object is close to the category boundary (dotted lines in the figure) or is an outlier, the attack performance will be poor. In this example, the `targeted attack success rate' of P2P and P2S is $33.3\%$ and $100\%$, respectively. }
	\label{fig:intro}
\end{figure}

Recent studies \cite{szegedy2013,goodfellow2014,xu2019exact,sapf2020,bai2020improving,chen2020boosting} revealed that DNNs are vulnerable to adversarial examples, which are crafted by adding intentionally small perturbations to benign examples and fool DNNs to confidently make incorrect predictions. While deep retrieval systems take advantage of the power of DNNs, they also inherit the vulnerability to adversarial examples \cite{feng2020universal,li2019universal,tolias2019targeted,yang2018adversarial}. Previous research \cite{yang2018adversarial} only paid attention to design a non-targeted attack in deep hashing based retrieval, \textit{i.e.}, returning retrieval objects with incorrect labels. Compared with non-targeted attacks, targeted attacks are more malicious since they make the adversarial examples misidentified as a predefined label and can be used to achieve some malicious purposes \cite{carlini2018audio,eykholt2018robust,qin2019imperceptible}. For example, a hashing based retrieval system may return violent images when a child queries with an intentionally perturbed cartoon image by the adversary. Accordingly, it is desirable to study the targeted attacks on deep hashing models and address their security concerns.

This paper focuses on the targeted attack in hashing based retrieval. Different from classification,  retrieval aims at returning multiple relevant objects instead of one result, which indicates that the query has more important relationship with the set of relevant objects than with other objects. Motivated by this fact, we formulate the targeted attack as a \emph{point-to-set} (P2S) optimization, which minimizes the average distance between the compressed representations (\textit{e.g.}, hash codes in Hamming space) of the adversarial example and those of a set of objects with the target label. Compared with the \emph{point-to-point} (P2P) paradigm \cite{tolias2019targeted} which directs the adversarial example to generate a representation similar to that of a randomly chosen object with the target label, our proposed point-to-set attack paradigm is more stable and efficient.
The detailed comparison between P2S and P2P attack paradigm is shown in Figure \ref{fig:intro}.
In particular, when minimizing the average Hamming distances between a hash code and those of an object set, we prove that the globally optimal solution (dubbed \emph{anchor code}) can be achieved through a simple component-voting scheme, which is a gift from the nature of hashing-based retrieval. Therefore, the anchor code can be naturally chosen as a targeted hash code to direct the generation of adversarial query.
To further balance the attack performance and the imperceptibility, we propose a novel attack method, dubbed \textit{\textbf{d}eep \textbf{h}ashing \textbf{t}argeted \textbf{a}ttack} (DHTA), by minimizing the Hamming distance between the hash code of adversarial query and the anchor code under the $\ell^\infty$ restriction on the adversarial perturbations.

In summary, the main contribution of this work is four-fold:

\begin{itemize}
    \item We formulate the targeted attack on hashing retrieval as a point-to-set optimization instead of the common point-to-point paradigm considering the characteristics of retrieval tasks.
    \item We propose a novel component-voting scheme to 
    obtain an anchor code as the representative of the set of hash codes of objects with the target label, whose theoretical optimality of proposed attack paradigm with average-case point-to-set metric is discussed.
    \item We develop a simple yet effective targeted attack, the DHTA, which efficiently balances the attack performance and the perceptibility. This is the first attempt to design a targeted attack on hashing based retrieval.
    \item Extensive experiments verify that DHTA is effective in attacking both deep hashing based image retrieval and video retrieval.
\end{itemize}

\section{Related Work}

\subsection{Deep Hashing based Similarity Retrieval}

Hashing methods can map semantically similar objects to similar compact binary codes in Hamming space, which are widely adopted to accelerate the ANN retrieval \cite{wang2017survey}. The classical version of data-dependent hashing consists of two parts, including hash function learning and binary inference \cite{li2012spectral,liu2012supervised,wang2015learning,shen2015supervised}.

Recently, more and more deep learning techniques were introduced to the traditional hashing-based retrieval methods and reach state-of-the-art performance, thanks to the powerful feature extraction of deep neural networks. The first deep hashing method was proposed in \cite{xia2014supervised} focusing on image retrieval. Recent works showed that learning hashing mapping in an end-to-end manner can greatly improve the quality of the binary codes \cite{lai2015simultaneous,liu2016deep,cao2017hashnet,cao2018deep}. The above-mentioned methods can be easily extended to multi-label image retrieval, $e.g.$, \cite{zhao2015deep,wu2017deep}. Depending on the availability of unlabeled images, other researchers devoted to design novel hashing methods to cope with the lack of labeled images, $e.g.$, unsupervised deep hashing method \cite{yang2019distillhash}, and semi-supervised one
\cite{yan2017semi}. Different from deep image hashing methods, deep video hashing usually first extract frame features by a convolutional neural network (CNN), then fuse them to learn global hashing function. Among various kinds of fusion methods, recurrent neural network (RNN) architecture is the most common choice, which can well model the temporal structure of videos \cite{gu2016supervised}. Moreover, some of the unsupervised video hashing methods were also proposed \cite{wu2018unsupervised,li2019neighborhood}, which organize the hash code learning in a self-taught manner to reduce the time and labor consuming labeling.

\subsection{Adversarial Attack}

DNNs can be easily fooled to confidently make incorrect predictions by intentional and human-imperceptible perturbations. The process of generating adversarial examples is called \emph{adversarial attack}, which was initially proposed by Szegedy $et \ al.$ \cite{szegedy2013} in the image classification task. To achieve such adversarial examples, the fast gradient sign method (FGSM) \cite{goodfellow2014} aims to maximize the loss along the gradient direction. After that, projected gradient descent (PGD) \cite{kurakin2016} was proposed to reach better performance. Deepfool finds the smallest perturbation by exploring the nearest decision boundary \cite{moosavi2016}. Except for the aforementioned attacks, many other methods \cite{carlini2017,dong2018,yao2019,wu2020skip} have also been developed to find the adversarial perturbation in the image classification problem. 

Besides, there are also other DNN based tasks that inherit the vulnerability to adversarial examples \cite{xu2019,dong2019efficient,ma2019understanding,duan2020adversarial}. Especially for the deep learning based similarity retrieval, it raises wide concerns on its security issues. For feature-based retrieval, Li $et \ al.$ \cite{li2019universal} focused on non-targeted attack by adding  universal adversarial perturbations (UAPs), while targeted mismatch adversarial attack was explored in \cite{tolias2019targeted}. In \cite{feng2020universal}, adversarial queries for deep product quantization network are generated by perturbing the overall soft-quantized distributions. However, for hashing based retrieval, one of the most important retrieval methods, its robustness analysis is left far behind. There is only one previous work in attacking deep hashing based retrieval \cite{yang2018adversarial}, which paid attention to the non-targeted attack, \textit{i.e.}, returning retrieval objects with the incorrect label. The targeted attack in such retrieval a system remains unaddressed.

\section{The Proposed Method}

\subsection{Preliminaries}

In this section, we briefly review the process of deep hashing based retrieval. Suppose $\bm{X}=\left \{ (\bm{x}_{i}, \bm{y}_{i}) \right \}_{i=1}^{N}$ indicates a set of $N$ sample collection labeled with $C$ classes, where $\bm{x}_{i}$ indicates the retrieval object, $e.g.$, a image or a video, and $\bm{y}_{i}\in\{0,1\}^C$  corresponds to a label vector. The $c$-th component of indicator vector  $\bm{y}_{i}^c=1$ means that the sample $\bm{x}_{i}$ belongs to class $c$. 
Let $\bm{X}^{(t)}=\left \{ (\bm{x}, \bm{y}) \in \bm{X}\mid  \bm{y} = \bm{y}_t \right \}$ be a subset of $\bm{X}$ consisting of those objects with label $\bm{y}_t$.

\noindent \textbf{Deep Hashing Model.}
The hash code of a query object $\bm{x}$ of deep hashing model is generated as follows:
\begin{equation}
\begin{split}
\bm{h}=F(\bm{x})={\rm sign}\left(f_{\theta}(\bm{x}) \right),
\end{split}
\label{eq:f_func}
\end{equation}
where $f_{\theta}(\cdot)$ is a DNN. In general, $f_{\theta}(\cdot)$ consists of a feature extractor followed by the fully-connected layers. Specifically, the feature extractor is usually specified by a CNN for image retrieval \cite{cao2017hashnet,cao2018deep,chen2019deep}, while CNN stacked with RNN is widely adopted for video retrieval \cite{gu2016supervised,song2018self,li2019neighborhood}.
In particular, the $\text{sign}(\cdot)$ function is approximated by the $\text{tanh}(\cdot)$ function during the training process in deep hashing based retrieval methods to alleviate the \emph{gradient vanishing} problem \cite{cao2017hashnet}.

\noindent \textbf{Similarity-based Retrieval.}
\comment{Given a code generation function $F(\cdot)$ and a object database $\{\bm{x}_i\}_{i=1}^N$, the similarity between the hash code of query $\bm{x}$ and that of each object $\bm{x}_i$ in database is denoted by $-d_H(F(\bm{x}), F(\bm{x}_i))$, which is calculated based on the Hamming distance $d_H(\cdot, \cdot)$. The retrieval process will return $M (1 \leq M \leq N)$ objects in database with top-$M$ similarity.}
Given a deep hashing model $F(\cdot)$, a query object $\bm{x}$ and a object database $\{\bm{x}_i\}_{i=1}^M$, the retrieval process is as follows. Firstly, the query $\bm{x}$ is fed into the deep hashing model and binary code $F(\bm{x})$ can be obtained through Eq. (\ref{eq:f_func}). Secondly, the Hamming distance between the hash code of query $\bm{x}$ and that of each object $\bm{x}_i$ in the database is calculated, denoted as $d_H(F(\bm{x}), F(\bm{x}_i))$. Finally, the retrieval system returns a list of objects, which is produced by sorting their Hamming distances.

\subsection{Deep Hashing Targeted Attack}

\noindent\textbf{Problem Formulation.}
In general, given a benign query $\bm{x}$, the objective of targeted attack in retrieval is to generate an attacked version $\bm{x}'$ of $\bm{x}$, which would cause the targeted model to retrieve objects with the target label $\bm{y}_t$. This objective can be achieved through minimizing the distance between the hash code of the attacked sample $\bm{x}'$ and those of the object subset $\bm{X}^{(t)}$ with the target label $\bm{y}_t$, $i.e.$,

\begin{equation}\label{obj_general}
    \min_{\bm{x}'} d\left(F(\bm{x}'\right), F(\bm{X}^{(t)})),
\end{equation}
where $F(\bm{X}^{(t)})=\{F(\bm{x})|\bm{x} \in \bm{X}^{(t)})\}$, and $d(\cdot,\cdot)$ denotes a point-to-set metric.

Once the problem is formulated as objective (\ref{obj_general}), the remaining problem is how to define the point-to-set metric. In this paper, we use the most widely used point-to-set metric, the \emph{average-case metric}, as shown in Definition \ref{defn_metrics}. 

\begin{definition}\label{defn_metrics}
Given a point $\bm{h}_0\in\{-1, +1\}^K$ and a set of points $\mathcal{A}$ in $\{-1, +1\}^K$ and point-to-point metric $d_H$, the average-case point-to-set metric is defined as follows: 
\begin{equation}\label{average}
d_{Ave}\left(\bm{h}_0, \mathcal{A}\right)\triangleq   \frac{1}{|\mathcal{A}|}\sum_{\bm{h}\in\mathcal{A}}d_H(\bm{h}_0, \bm{h}).
\end{equation}

\end{definition}

\begin{remark}
If average-case point-to-set metric is adopted, the objective function (\ref{obj_general}) is specified as  
\begin{equation}\label{obj_ave}
\min_{\bm{h}'}\frac{1}{|\mathcal{A}|} \sum_{\bm{h} \in F(\bm{X}^{(t)})}\,d_H(\bm{h}', \bm{h}),
\end{equation}
where $\bm{h}'$ is the hash code corresponding to the adversarial example $\bm{x}'$.
\end{remark}

In particular, there exists an analytical optimal solution (dubbed \emph{anchor code}) of the optimization problem (\ref{obj_ave}) obtained through a \emph{component-voting scheme},  which is a property arising from the nature of Hamming distance of hashing-based retrieval. The \emph{component-voting scheme} is shown in Algorithm \ref{alg_vote}, and the optimality of anchor code is verified in Theorem \ref{thm}. The proof is shown in the \textbf{Appendix}.

\begin{theorem} \label{thm}
 Anchor code ${\bm{h}_a}$ calculated by Algorithm \ref{alg_vote} is the binary code achieving the minimal sum of Hamming distances with respect to ${\bm{h}_i}$, $i = 1, \dots, n_t$, i.e., 
 \begin{equation}
 {\bm{h}_a}=\arg \min_{\bm{h}\in\{+1,-1\}^K}\sum_{i=1}^{n_t}d_H({\bm{h}}, {\bm{h}_i}).
 \end{equation}
\end{theorem}

\begin{algorithm}[!hbtp]
\caption{Component-voting Scheme} 
{\bf Input:} 
$K$-bits hash codes $\{\bm{h}_i\}_{i=1}^{n_t}$ of objects with the target label $t$.\\
{\bf Output:}  
Anchor code $\bm{h}_a$.
\begin{algorithmic}[1]
\For{$j=1 : K$}
\State Conduct voting process through counting up the number of $+1$ and $-1$, denoted by $N_{+1}^j$ and $N_{-1}^j$, respectively. For the $j$-th component among $\{\bm{h}_i\}_{i=1}^{n_t}$, $i.e.$, 
\begin{equation}
N_{+1}^j=\sum_{i}^{n_t} \mathbb I (\bm{h}_{i}^{j}=+1),\qquad N_{-1}^j=\sum_{i=1}^{n_t} \mathbb I (\bm{h}_{i}^{j}=-1),
\end{equation}
where $\mathbb I(\cdot)$ is an indicator function.

\State Determine the $j$-th component of anchor code $\bm{h}_{a}^{j}$ as 
\begin{equation}
\begin{split}
\bm{h}_{a}^{j}=\begin{cases}
+1, & \text{ if } N_{+1}^j \geqslant N_{-1}^j \\ 
-1, & \text{ otherwise } 
\end{cases}
\end{split}
.
\label{eq:ahc}
\end{equation}

\EndFor
\State \Return Anchor code $\bm{h}_a$.
\end{algorithmic}
\label{alg_vote}
\end{algorithm}

\comment{
\begin{proof}
We only need to prove that for any ${{\bm{h}}\in\{+1,-1\}}^K$ and ${\bm{ h}}\neq {\bm{h}_a}$, the following inequality holds.
\begin{equation}
\sum_{i}^{n_t}d_H({\bm{h}_a}, {\bm{h}_i})\leq \sum_{i}^{n_t}d_H({\bm{h}}, {\bm{h}_i}).
\end{equation}
Denote $\mathcal{D}=\{j_1, j_2,\dots, j_{K_0}\}$, $1\leq K_0\leq K$, as the index set where  ${\bm{h}}$ and  ${\bm{h}_a}$ differ. Then we have 
\begin{eqnarray}
&\quad&\sum_{i}^{n_t}d_H({\bm{h}_a}, {\bm{ h}_i})\nonumber\\
&=&\sum_{j\in \mathcal{D}}d_H({\bm{h}_a^j}, {\bm{h}_i^j})+\sum_{j\in \{1, 2, \dots, K\}\setminus\mathcal{D}}d_H({\bm{h}_a^j}, {\bm{h}_i^j})\\
&=&\sum_{j\in \mathcal{D}}\big(n_t-\sum_{i}^{n_t}\mathbb{I}({\bm{h}_a^j}={\bm{h}_i^j})\big)+\sum_{j\in \{1, 2, \dots, K\}\setminus\mathcal{D}}\big(n_t-\sum_{i}^{n_t}\mathbb{I}({\bm{ h}_a^j}={\bm{ h}_i^j})\big)\\
&\overset{(a)}{\leq}& \sum_{j\in \mathcal{D}}\big(n_t-\sum_{i}^{n_t}\mathbb{I}({\bm{h}^j}={\bm{h}_i^j})\big)+\sum_{j\in \{1, 2, \dots, K\}\setminus\mathcal{D}}\big(n_t-\sum_{i}^{n_t}\mathbb{I}({\bm{h}^j}={\bm{h}_i^j})\big)\\
&=&\sum_{i}^{n_t}d_H({\bm{h}}, {\bm{h}_i}),
\end{eqnarray}
where $(a)$ holds since anchor code $\bm{h}_a$ is obtained through a voting process (as shown in Algorithm \ref{alg_vote}), $i.e.$,  $ \forall j\in \mathcal{D}$, 
\begin{equation}
    \sum_{i=1}^{n_t} \mathbb{I}({\bm{h}_a^j}={\bm{h}_i^j})\geq \sum_{i=1}^{n_t} \mathbb{I}({\bm{h}^j}={\bm{h}_i^j}).
\end{equation}

$\hfill\blacksquare$ 

\end{proof}
}

\begin{figure}[ht]
	\centering
	\includegraphics[width=0.95\textwidth]{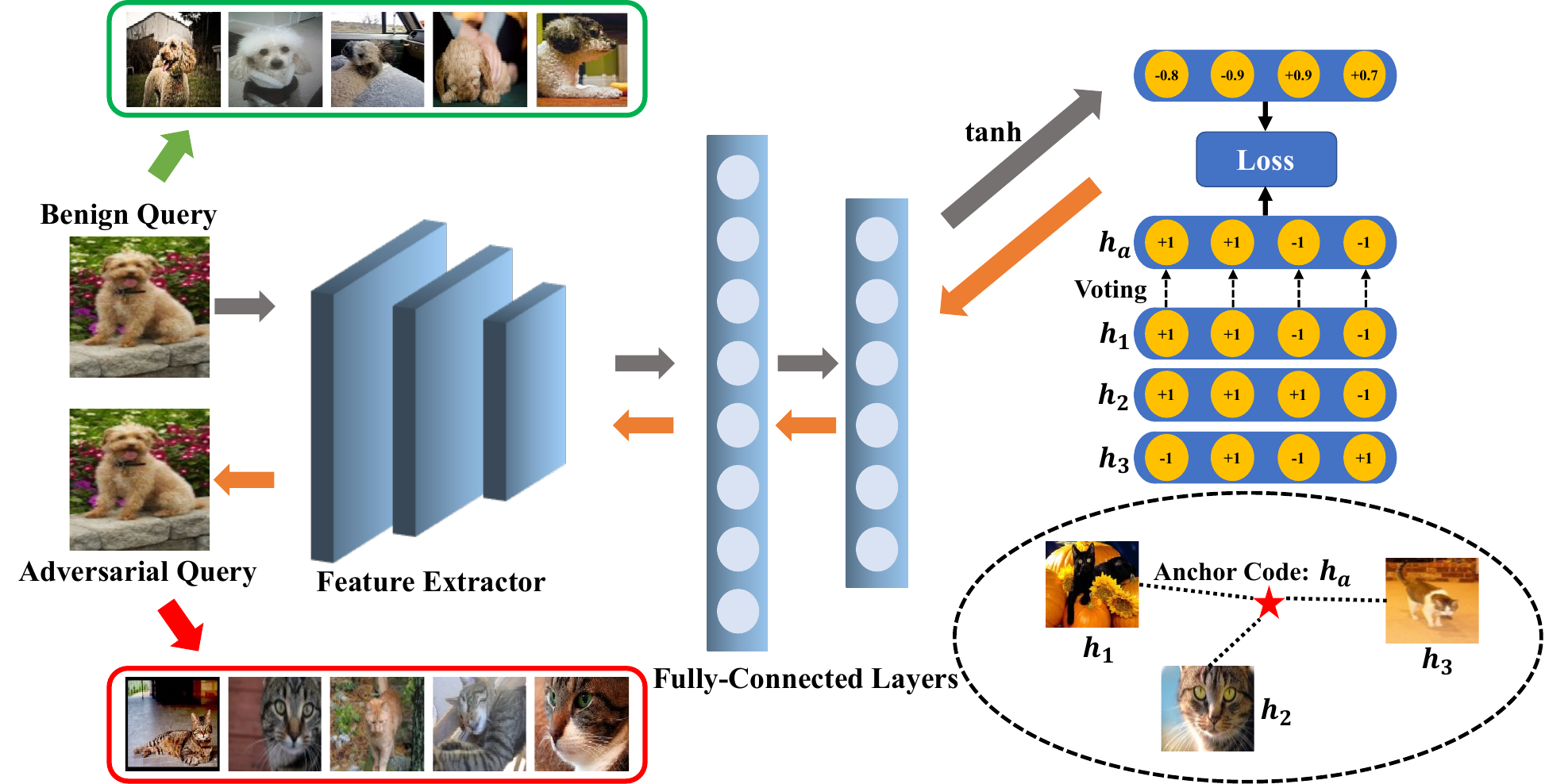}
	\caption{The pipeline of proposed DHTA method, where the gray and orange arrows indicate forward and backward propagation, respectively. The adversarial query is generated through minimizing the loss calculated by its hash code and the anchor code of the set of objects with the target label. The anchor code $h_a$ is calculated through the component-voting scheme ($i.e.$ an entry-wise voting process). In this toy example, $\bm{h}_1$, $\bm{h}_2$ and $\bm{h}_3$ are three 4 bits hash codes of objects with the target label ``Cat''.}
	\label{fig_method}
\end{figure}

\noindent \textbf{Overall Objective Function.}
Due to the optimal representative property of anchor code for the set of hash codes of objects with the target label (Theorem 1), we can naturally choose the anchor code as a targeted hash code to direct the generation of the adversarial query. However, the attacked object corresponding to the anchor code may be far different from the original one visually, which would cause the attacked object to be easily detectable. To solve this problem, we introduce the $\ell^{\infty}$ restriction on the adversarial perturbations 
by minimizing the Hamming distance between the hash code of attacked object and that of the anchor code as follows:

\begin{equation}
\begin{split}
\underset{\bm{x}^\prime}{\text{ min }}  d_H(\text{sign}(f_{\theta}(\bm{x}^\prime)), \bm{h}_a) \ \ \
s.t. \ \ ||\bm{x}^\prime-\bm{x}||_\infty \leq\epsilon,
\end{split}
\label{eq:loss_obj}
\end{equation}
where $\epsilon$ denotes the maximum perturbation strength, $\bm{h}_a$ is the anchor code of object set with the target label.

Besides, given a pair of binary codes $\bm{h}_i$ and $\bm{h}_j$, since $d_{H}(\bm{h}_{i}, \bm{h}_{j})=\frac{1}{2}(K-\bm{h}_i^{\top}\bm{h}_j)$, we can equivalently replace Hamming distance with inner product in the objective function. In particular, similar to deep hashing methods \cite{cao2017hashnet}, we adopt the hyperbolic tangent ($\text{tanh}$) function to approximate sign function for the adversarial generation. Similar to \cite{yang2018adversarial}, we also introduce the factor $\alpha$ to address the \emph{gradient vanishing} problem. In summary, the overall optimization objective of proposed method is as follows:

\begin{equation}
\begin{split}
\underset{\bm{x}^\prime}{\text{ min }}  -\frac{1}{K}\bm{h}_{a}^{\top}{\rm tanh}(\alpha f_{\theta}(\bm{x}^\prime)) \ \ \
s.t. \ \ ||\bm{x}^\prime-\bm{x}||_\infty \leq\epsilon,
\end{split}
\label{eq:loss_final}
\end{equation}
where the hyper-parameter $\alpha \in [0,1]$, $\bm{h}_a$ is the anchor code.

The overall process of proposed DHTA is shown in Figure \ref{fig_method}.

\section{Experiments}

\subsection{Benchmark Datasets and Evaluation Metrics}
\label{sec:datasets_metric}

Four benchmark datasets, including ImageNet \cite{russakovsky2015imagenet}, NUS-WIDE \cite{chua2009nus}, JHMDB \cite{jhuang2013towards}, and UCF-101 \cite{soomro2012ucf101}, are adopted in our experiments. The first two datasets are used for image retrieval, while the last two are used for video retrieval. More details about datasets are described in the \textbf{Appendix}.

For the evaluation of targeted attacks, we define the targeted mean average precision (t-MAP) as the evaluation metric, which is similar to mean average precision (MAP) widely used in information retrieval \cite{zuva2012evaluation}. Specifically,  the referenced label of t-MAP is the targeted label instead of the original one of the query object in MAP. The higher the t-MAP, the better the targeted attack performance. In image hashing, we evaluate t-MAP on top 5,000 and 1,000 retrieved images on NUS-WIDE and ImageNet, respectively. We evaluate t-MAP on all retrieved videos in video hashing. Besides, we also present the precision-recall curves (PR curves) of different methods for more comprehensive comparison.

\begin{table}[ht]
\centering
\caption{t-MAP (\%) of targeted attack methods and MAP (\%) of query with benign objects (‘Original’) with various code lengths on two image datasets.}
\setlength{\tabcolsep}{1mm}{
	\begin{tabular}{lcccccccccc}
		\toprule
		\multirow{2}{*}{Method} & \multirow{2}{*}{Metric} & \multicolumn{4}{c}{ImageNet} & & \multicolumn{4}{c}{NUS-WIDE} \\ \cline{3-6} \cline{8-11}
		& & 16bits & 32bits & 48bits & 64bits & & 16bits & 32bits & 48bits & 64bits \\ \midrule
		Original &t-MAP  & 3.80& 1.36& 1.64& 1.98& &37.62& 36.03& 38.32& 38.69  \\
		Noise    &t-MAP  & 3.29& 1.24& 1.89& 2.10& &37.34& 36.15& 38.25& 38.57  \\
		P2P      &t-MAP  & 44.35& 58.32& 62.50& 65.61& &75.45& 78.59& 81.40& 81.28  \\
		DHTA     &t-MAP  &  \textbf{63.68}&  \textbf{77.76}&  \textbf{82.31}&  \textbf{82.10}& & \textbf{82.35}&  \textbf{85.66}&  \textbf{86.80}&  \textbf{88.84}  \\ \midrule
		Original &MAP    & 51.02& 62.70& 67.80& 70.11& &76.93& 80.37& 82.06& 81.62  \\ \bottomrule
	\end{tabular}
}

\label{tab:image_map}
\end{table}

\begin{figure}[ht]
	\centering
	\includegraphics[width=1\textwidth]{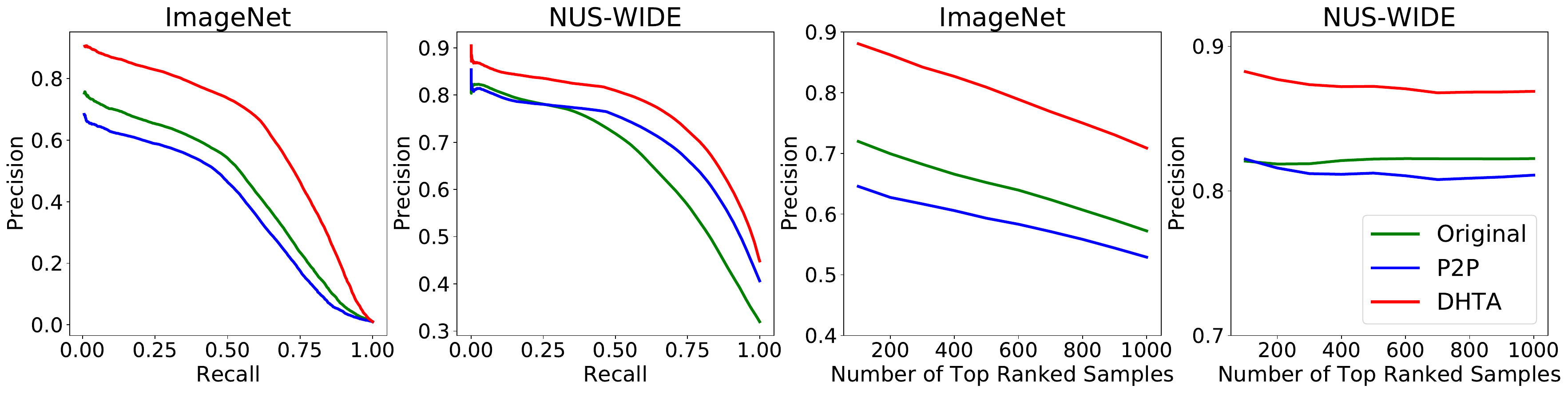}
	\caption{Precision-recall and precision curves under 48 bits code length in image retrieval. P2P attack and DHTA are evaluated based on the target label, while the result of `Original' is calculated based on the label of query object.}
	\label{fig:image_pr_prec}
\end{figure}

\begin{figure}[ht]
	\centering
	\includegraphics[width=0.95\textwidth]{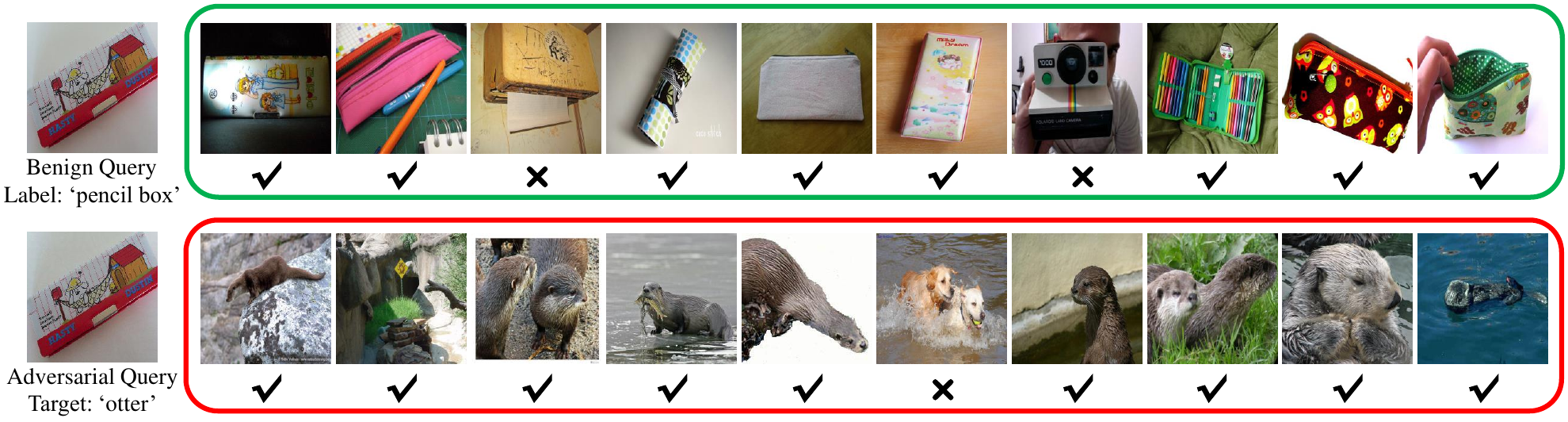}
	\caption{An example of image retrieval with benign query and its adversarial query on ImageNet. Retrieved objects with top-10 similarity are shown in the box. The tick and cross indicate whether the retrieved object is consistent with the desired label (the original label for benign query and the target label for adversarial query).}
	\label{fig:img_exp}
\end{figure}

\subsection{Overall Results on Image Retrieval}
\label{sec:image_setting}

\noindent \textbf{Evaluation Setup.}
For image hashing, we adopt VGG-11 \cite{simonyan2014very} as the backbone network pre-trained on ImageNet to extract features, then replace the last fully-connected layer of softmax classifier with the hashing layer. The detailed settings of training image hashing models are illustrated in the \textbf{Appendix}. For each dataset, we randomly select 100 samples from the query set as benign queries to evaluate the performance of attack. For each generation, we randomly select a label as the target label different from the label of query. When generating an anchor code, we randomly sample images from objects in the database with the target label to form the hash code set. For all adversarial examples, the perturbation magnitude $\epsilon$ of normalized data and $n_t$ is set to 0.032 and 9, respectively. We adopt PGD \cite{kurakin2016} to optimize the proposed attack. We attack image hashing models with learning rate 1 and the number of iterations is set to 2,000. Following \cite{yang2018adversarial}, the parameter $\alpha$ is set as 0.1 during the first 1,000 iterations, and is updated every 200 iterations according to [0.2, 0.3, 0.5, 0.7, 1.0] during the last 1,000 iterations. We compare DHTA with targeted attack with P2P paradigm \cite{tolias2019targeted}, which is specified as DHTA with $n_t=1$. We also show the t-MAP results of images with additive noise sampled from the uniform distribution $U(-\epsilon, +\epsilon)$.

\noindent \textbf{Results.}
The general attack performance of different methods is shown in Table \ref{tab:image_map}. The t-MAP values of query with benign objects (dubbed \emph{Original}) or query with noisy objects (dubbed \emph{Noise}) are relatively small on both ImageNet and NUS-WIDE datasets. Especially on ImageNet dataset, the t-MAP values of two aforementioned methods are closed to 0, which indicates that query with benign images or images with noise can not successfully retrieve objects with the target labels as expected. In contrast, designed targeted attack methods ($i.e.$ P2P and DHTA) can significantly improve the t-MAP values. For example, compared with the t-MAP of benign query on ImageNet dataset, the improvement of P2P methods is over $40\%$ in all cases. Especially under the relatively large code length (64 bits), the improvement even goes to $63\%$. Among two targeted attack methods, the proposed DHTA method achieves the best performance. Compared with P2P, the t-MAP improvement of DHTA is over $16\%$ (usually over $19\%$) in all cases on the ImageNet dataset. Moreover, the t-MAP values of targeted attacks increase as the number of bits, which is probably caused by the extra information introduced in the longer code length. In particular, an interesting phenomenon is that the t-MAP value of DHTA is even significantly higher than the MAP value of `Original', which suggests that the attack performance of DHTA is not hindered by the performance of the original hashing model ($i.e.$ threat model) to some extent. An example of the results of query with a benign image and an adversarial image is displayed in Figure \ref{fig:img_exp}. 

Furthermore, we also provide the precision-recall and  precision curves for a more comprehensive comparison. As shown in Figure \ref{fig:image_pr_prec}, the curves of DHTA are always above all other curves, which demonstrates that the performance of DHTA does better than all other methods.

\begin{table}[t]
	\centering
	\caption{t-MAP (\%) of targeted attack methods and MAP (\%) of query with benign objects (‘Original’) with various code lengths on two video datasets.}
		\setlength{\tabcolsep}{1mm}{
			\begin{tabular}{lcccccccccc}
				\toprule
				\multirow{2}{*}{Method} & \multirow{2}{*}{Metric} & \multicolumn{4}{c}{JHMDB} & & \multicolumn{4}{c}{UCF-101} \\ \cline{3-6} \cline{8-11}
				& & 16bits & 32bits & 48bits & 64bits & & 16bits & 32bits & 48bits & 64bits \\ \midrule
				Original &t-MAP  & 6.73& 6.26& 6.48& 6.89& &1.69& 1.67& 1.79& 1.86  \\
				Noise    &t-MAP  & 6.67& 6.13& 6.50& 6.94& &1.69& 1.72& 1.87& 1.85  \\
				P2P      &t-MAP  & 39.67& 42.37& 44.78& 44.38& &55.57& 53.49& 55.27& 51.88  \\
				DHTA     &t-MAP  &  \textbf{56.47}&  \textbf{62.04}&  \textbf{63.02}&  \textbf{66.06}& & \textbf{67.84}&  \textbf{66.18}&  \textbf{69.72}&  \textbf{67.83}  \\ \midrule
				Original &MAP    & 35.18& 42.46& 45.80& 45.50& &55.16& 55.25& 56.56& 56.79  \\ \bottomrule
			\end{tabular} 
		}
	\label{tab:video_map}
\end{table}

\begin{figure}[t]
	\centering
	\includegraphics[width=1\textwidth]{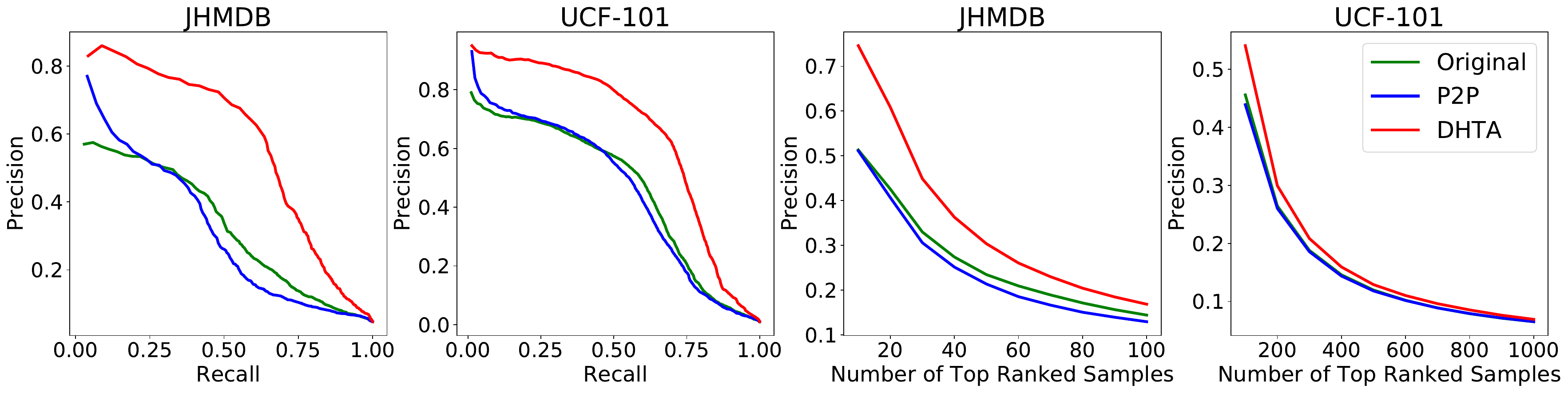}
	\caption{Precision-recall and precision curves under 48 bits code length in video retrieval. P2P attack and DHTA are evaluated based on the target label, while the result of `Original' method is calculated based on the label of query object.}
	\label{fig:video_pr_prec}
\end{figure}

\begin{figure}[ht]
	\centering
	\includegraphics[width=0.99\textwidth]{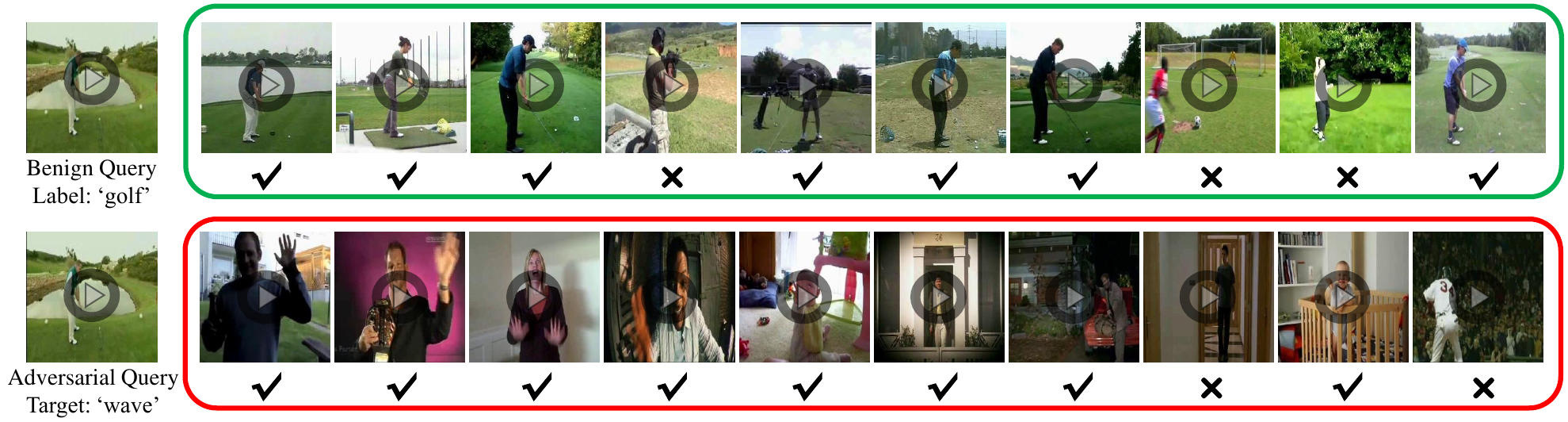}
	\caption{An example of video retrieval with benign query and its adversarial query on JHMDB. Retrieved objects with top-10 similarity are shown in the box. The tick and cross indicate whether the retrieved object is consistent with the desired label (the original label for benign query and the target label for adversarial query).
	}
	\label{fig:video_exp}
\end{figure}

\subsection{Overall Results on Video Retrieval}
\label{sec:video_setting}

\noindent \textbf{Evaluation Setup.}
According to model architectures of the state-of-the-art deep video retrieval methods \cite{gu2016supervised,song2018self,li2019neighborhood}, we adopt AlexNet \cite{krizhevsky2012imagenet} to extract spatial features and LSTM \cite{hochreiter1997long} to fuse the temporal information. The detailed settings of training video hashing model are presented in the \textbf{Appendix}. For attacking video hashing, the number of iterations is 500, and the parameter $\alpha$ is fixed at 0.1. Other settings are the same as those used in Section \ref{sec:image_setting}.

\begin{figure}[ht]
    \centering
    \subfigure{
    \includegraphics[width=0.485\textwidth]{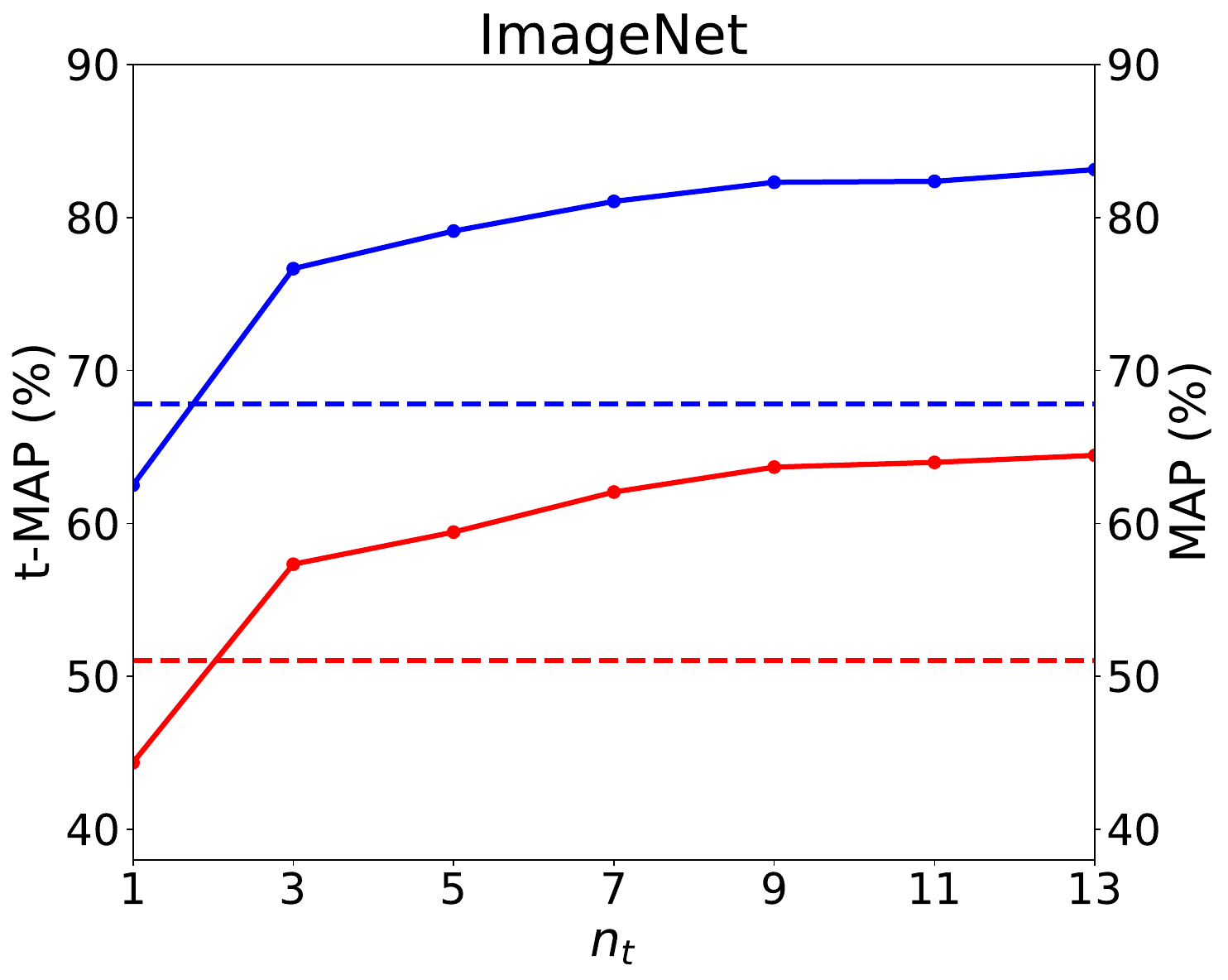}}
    \subfigure{
    \includegraphics[width=0.485\textwidth]{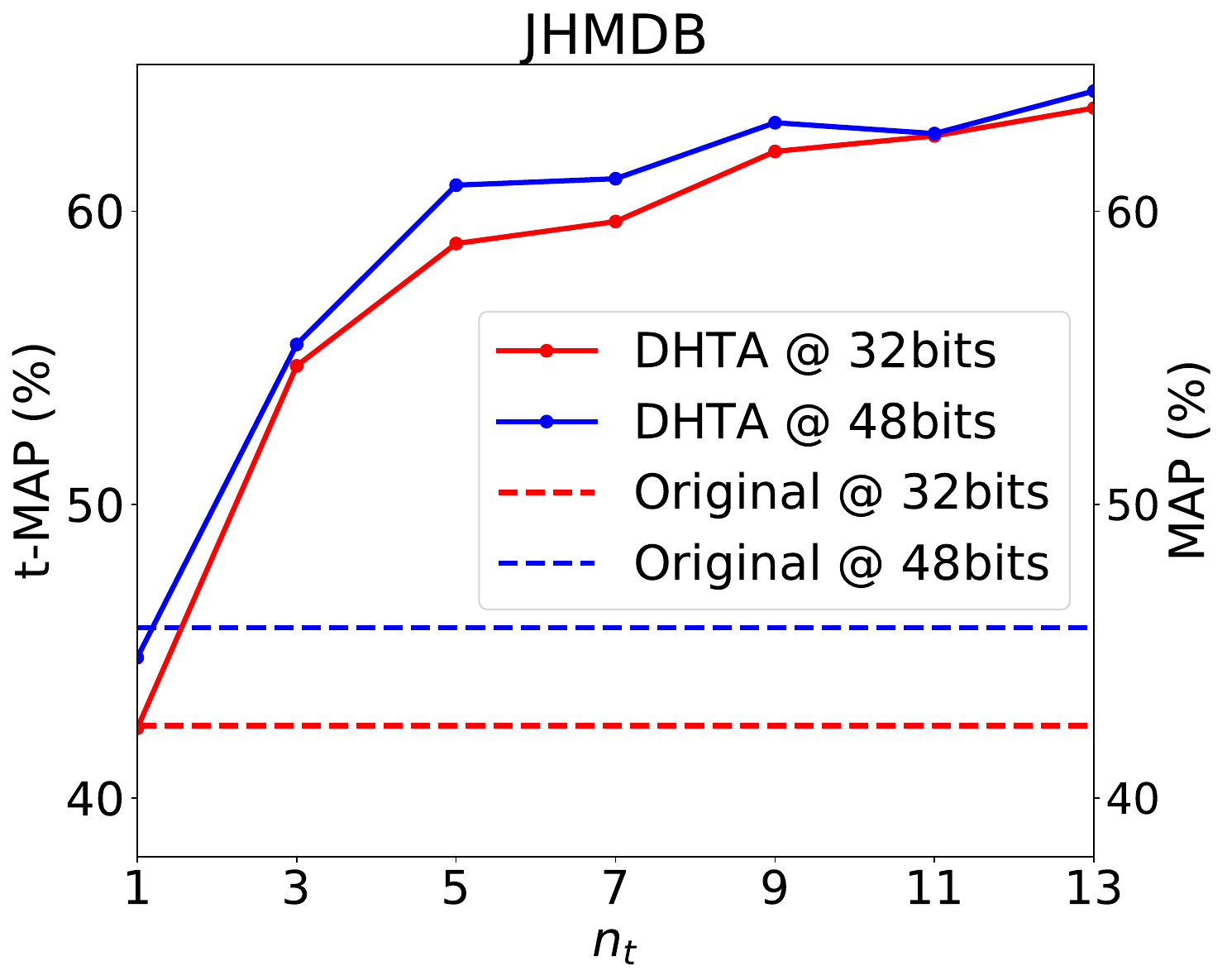}}
    \caption{t-MAP (\%) of DHTA and MAP (\%) of query with benign objects (`Original') with different $n_t$ and code length on ImageNet and JHMDB.}
	\label{fig:nt}
\end{figure}

\noindent \textbf{Results.}
The attack performance in video retrieval is shown in Table \ref{tab:video_map}. Similar to the image scenario, query with benign videos or videos with noise can not successfully retrieve objects with the target label, thus fails to attack the deep hashing based retrieval. In contrast, deep hashing based video retrieval can be easily attacked by designed targeted attacks, especially the DHTA proposed in this paper. For example, the t-MAP value of DHTA is $59\%$ over query with benign videos, and $21\%$ over P2P attack paradigm on the JHMDB dataset with code length 64 bits. The precision-recall and the precision curves also verify the superiority of DHTA over other methods, as shown in Figure \ref{fig:video_pr_prec}. Especially on JHMDB dataset, there exists a significantly large gap between the PR curve of DHTA and those of other methods. In addition, the t-MAP value of DHTA is again significantly larger than the MAP of the benign query (the `Original'). An example of the results of query with a benign video and an adversarial video is displayed in  Figure \ref{fig:video_exp}.

\subsection{Discussion}

\noindent \textbf{Effect of $\bm{n_t}$.}
To analyze the effect of the size of object set for generating the anchor code ($i.e.,$ $n_t$), we discuss the t-MAP of DHTA under different values of $n_t \in \{1, 3, 5, 7, 9, 11, 13\}$. Other settings are the same as those used in Section \ref{sec:image_setting}-\ref{sec:video_setting}. We use ImageNet and JHMDB as the representative for analysis.  

As shown in Figure \ref{fig:nt}, the t-MAP value increase as the increase of $n_t$ under different code lengths. The MAP of corresponding query with benign objects ($i.e.$ the `Original') can be regarded as the reference of the retrieval performance. We observe that the t-MAP is higher than the MAP of its corresponding `Original' method in all cases when $n_t \geq 3$. In other words, DHTA can still have satisfying performance with relatively small $n_t$. This advantage is critical for attackers, since the bigger the $n_t$, the higher the cost of data collection and adversarial generation for an attack. It is worth noting that the attack performance degrades significantly when $n_t=1$, which exactly corresponds to the P2P attack paradigm.

\noindent \textbf{Effect of the Number of Iterations.}
Table \ref{tab:iters_image}-\ref{tab:iters_video} present the t-MAP of DHTA with different iterations on ImageNet and JHMDB datasets. Except for the iterations, other settings are the same as those used in Section \ref{sec:image_setting}-\ref{sec:video_setting}.

As expected, the t-MAP values increase with the number of iterations. Even with relatively few iterations, the proposed DHTA can still achieve satisfying performance. For example, with 100 iterations, the t-MAP values are over $50\%$ under all code lengths. Especially on the ImageNet dataset, the t-MAP is over $70\%$ with relatively larger code length ($\geq$48 bits). These results consistently verify the high-efficiency of our DHTA method.

\begin{table}[htbp]
  \begin{minipage}{0.47\linewidth}
    \centering
		\caption{t-MAP (\%) of DHTA with different iterations on ImageNet.}
			\begin{tabular}{lccccc}
				\toprule
				Iteration& 16bits & 32bits & 48bits & 64bits \\ \midrule
				100  & 52.99& 66.29& 70.65& 72.43\\
				500  & 55.18& 68.30& 74.47& 76.15\\
				1000 & 56.96& 68.36& 75.03& 76.25\\
				1500 & 62.81& 74.11& 79.28& 78.71\\ 
				2000 & 63.68& 77.76& 82.31& 82.10\\ \bottomrule
			\end{tabular}
	        \label{tab:iters_image}
  \end{minipage}
  \quad
  \begin{minipage}{0.47\linewidth}
    \centering
		\caption{t-MAP (\%) of DHTA with different iterations on JHMDB.}
			\begin{tabular}{lccccc}
				\toprule
				Iteration& 16bits & 32bits & 48bits & 64bits \\ \midrule
				10  & 28.51& 23.88& 22.84& 23.21\\
				50  & 48.69& 48.18& 47.01& 48.97\\
				100 & 53.21& 54.91& 55.94& 58.28\\
				500 & 56.47& 62.04& 63.02& 66.06\\ \bottomrule
			\end{tabular}
	        \label{tab:iters_video}
  \end{minipage}
\end{table}

\begin{table}[ht]
	\centering
	\caption{MAP (\%) of different methods on ImageNet and JHMDB. The best results are marked with boldface, while the second best results are marked with underline.}
	\setlength{\tabcolsep}{1.3mm}{
		\begin{tabular}{lcccccccccc}
            \toprule
			\multirow{2}{*}{Method}  & \multicolumn{4}{c}{ImageNet} & & \multicolumn{4}{c}{JHMDB} \\ \cline{2-5} \cline{7-10}
            & 16bits & 32bits & 48bits & 64bits & & 16bits & 32bits & 48bits & 64bits \\ \midrule
            Original   & 51.02& 62.70& 67.80& 70.11 && 35.18&  42.46&	45.80&	45.50&\\
            Noise      & 50.94& 62.52& 66.69& 69.85 && 35.04&	42.15&	45.67&	45.63&\\
            P2P        & 3.36 & \textbf{2.48} & \underline{2.45} & 3.93 && 7.71&	8.20&	8.14&	10.19&\\
            HAG &  \underline{1.88} & 4.96 &  3.89 &  \underline{2.34} && \textbf{3.52}&	\textbf{3.58}&	\textbf{3.42}&	\textbf{3.34}&\\
            DHTA       & \textbf{0.54} &  \underline{5.64} & \textbf{2.30} & \textbf{1.70} &&  \underline{6.76}&	 \underline{7.23}&	 \underline{6.56}&	 \underline{7.55}&\\  \bottomrule 
		\end{tabular}
	}
	\label{tab:untarget}
\end{table}

\noindent \textbf{Evaluation from the Perspective of Non-targeted Attack.}
Targeted attack can be regarded as a special non-targeted attack, since the target label is usually different from the one of query object. In this part, we compare the targeted attacks (P2P and DHTA) with other methods, including additive noise and HAG \cite{yang2018adversarial} (which is the state-of-the-art non-targeted attack), in the non-targeted attack scenario. 

The MAP results of different methods are reported in Table \ref{tab:untarget}. The lower the MAP, the better the non-targeted attack performance. As shown in the table, although targeted attacks are not designed for the non-targeted scenario, they still have competitive performance. For example, the MAP values of DHTA are $50\%$ smaller than those of `Original' under all code length on ImageNet. Especially for the proposed DHTA, it even has better non-targeted attack performance ($i.e.$ smaller MAP) compared with HAG on ImageNet in most cases.

\begin{figure}[ht]
\centering
\subfigure[ImageNet]{
\includegraphics[width=0.99\textwidth]{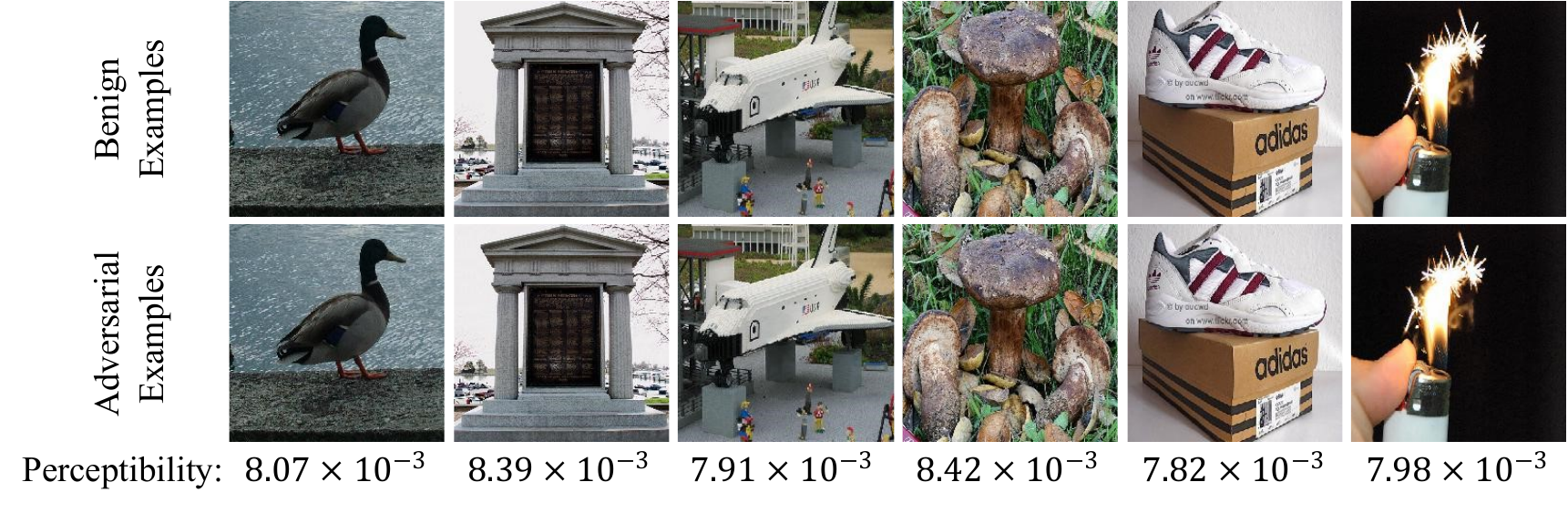}}
\subfigure[NUS-WIDE]{
\includegraphics[width=0.99\textwidth]{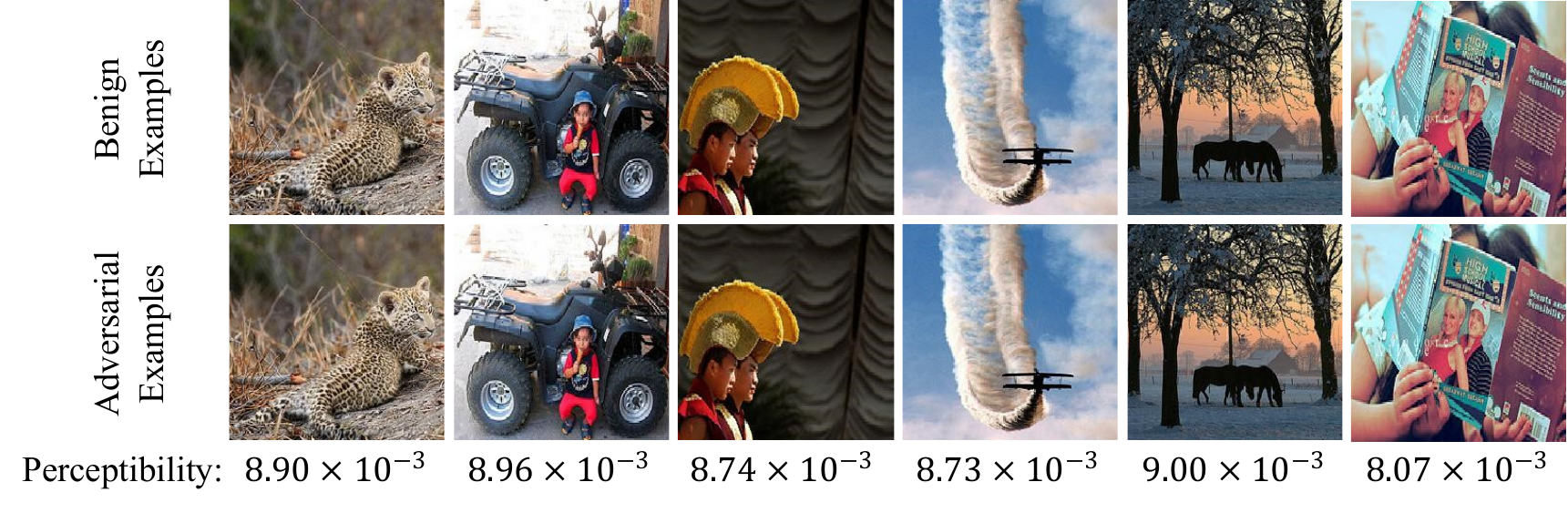}}
\caption{Visualization examples of generated adversarial examples in image hashing.}
\label{fig:adv_exp}
\end{figure}

\noindent \textbf{Perceptibility.}
Except for the attack performance, the \emph{perceptibility} of adversarial perturbations is also important. Following the setting suggested in \cite{szegedy2013,tramer2017ensemble}, given a benign query $\bm{x}$, the perceptibility of its corresponding adversarial query $\bm{x}'$ is defined as $\sqrt{\frac{1}{n}\left\|\bm{x}'-\bm{x}\right\|_2^2}$, where $n$ is the size of the object and pixel values are scaled to be in the range [0, 1].

For each dataset, we calculate the average perceptibility over all generated adversarial objects. The perceptibility value of ImageNet and NUS-WIDE datasets is $8.35\times 10^{-3}$ and $9.07\times 10^{-3}$, respectively. In video retrieval tasks, the value is $5.81\times 10^{-3}$ and $7.72\times 10^{-3}$ on JHMDB and UCF-101 datasets, respectively. These results indicate that the adversarial queries are very similar to their original versions. Some adversarial images are shown in Figure \ref{fig:adv_exp}, while examples of video retrieval are shown in the \textbf{Appendix}.

\subsection{Open-set Targeted Attack}

\noindent \textbf{Evaluation  Setup.}
In the above experiments, the target label is selected from those of training set. In this section, we use ImageNet dataset as an example to further evaluate the proposed DHTA under a tougher open-set scenario, where the out-of-sample class will be assigned as the target label. This setting is more realistic since the attacker may probably not be able to access the training set of the attacked deep hashing model. For example, the deep hashing model may be downloaded from a third-party open-source platform where the training set is unavailable. 

\begin{table}[t]
\centering
\caption{t-MAP (\%) of DHTA with out-of-sample target label on ImageNet.}
\setlength{\tabcolsep}{1mm}{
	\begin{tabular}{lcccc}
		\toprule
		\multicolumn{1}{c}{Method} & 16bits & 32bits & 48bits & 64bits \\ \midrule
        DHTA ($n_t=5$)  & 33.67& 46.34& 48.91& 48.27  \\
		DHTA ($n_t=7$)  & 34.77& 50.92& 51.68& 49.18  \\
		DHTA ($n_t=9$)  & 37.34& 54.13& 55.12& 52.17  \\
		DHTA ($n_t=11$) & 38.00& 54.05& 56.93& 54.12  \\ 
		\bottomrule
	\end{tabular}
}
\label{fig:unseen}
\end{table}

Specifically, we randomly select 10 additional classes different from those used for training a deep hashing model in Section \ref{sec:datasets_metric}. These 
selected images from 10 additional classes will be treated as an open set for our evaluation. When generating the anchor code of objects with the target label (within the open set), we remain our deep hashing model trained on the previous 100 classes.

\noindent \textbf{Results.}
As shown in Table \ref{fig:unseen}, DHTA still has a certain attack effect even if the target label is out-of-sample. Especially when the $n_t$ and the code length tend larger, the t-MAP values of DHTA are over $50\%$. This phenomenon may reveal that the learned feature extractor did learn some useful low-level features, which represents those objects with the same class in some similar locations in Hamming space, no matter the class is learned or not. In addition, the attack performance is also increasing with the $n_t$ and code length. 

\section{Conclusion and Future Work}

In this paper, we explore the landscape of the targeted attack for deep hashing based retrieval. Based on the characteristics of the retrieval task, we formulate the attack as a point-to-set optimization, which minimizes the average distance between the hash code of the adversarial example and those of a set of objects with the target label. Theoretically, we propose a component-voting scheme to obtain the optimal representative, the anchor code, for the code set of point-to-set optimization. Based on the anchor code, we propose a novel targeted attack method, the DHTA, to balance the performance and perceptibility through minimizing the Hamming distance between the hash code of adversarial example and the anchor code under the $\ell^\infty$ restriction on the adversarial perturbation. Extensive experiments are conducted, which verifies the effectiveness of DHTA in attacking both deep hashing based image retrieval and video retrieval. To alleviate the proposed threat, we will discuss how to generalize existing adversarial training based methods from P2P to the P2S scheme for the defense. The specific approaches will be further demonstrated in our future works.

\section*{Acknowledgments}

This work is supported in part by the National Key Research and Development Program of China under Grant 2018YFB1800204, the National Natural Science Foundation of China under Grant 61771273, the R\&D Program of Shenzhen under Grant JCYJ20180508152204044, the project ``PCL Future Greater-Bay Area Network Facilities for Large-scale Experiments and Applications (LZC0019)'', the Natutal Sciences and Engineering Research Council of Canada under Grant RGPIN203035-16, and the Canada Research Chairs Program. We also thank vivo and Rejoice Sport Tech. co., LTD. for their GPUs. 

\clearpage
%
%
\bibliographystyle{splncs04}
\bibliography{egbib}

\begin{thebibliography}{10}
\providecommand{\url}[1]{\texttt{#1}}
\providecommand{\urlprefix}{URL }
\providecommand{\doi}[1]{https://doi.org/#1}

\bibitem{bai2020improving}
Bai, Y., Zeng, Y., Jiang, Y., Wang, Y., Xia, S.T., Guo, W.: Improving query
  efficiency of black-box adversarial attack. In: ECCV (2020)

\bibitem{cao2018deep}
Cao, Y., Long, M., Liu, B., Wang, J.: Deep cauchy hashing for hamming space
  retrieval. In: CVPR (2018)

\bibitem{cao2017hashnet}
Cao, Z., Long, M., Wang, J., Yu, P.S.: Hashnet: Deep learning to hash by
  continuation. In: ICCV (2017)

\bibitem{carlini2017}
Carlini, N., Wagner, D.: Towards evaluating the robustness of neural networks.
  In: IEEE S\&P (2017)

\bibitem{carlini2018audio}
Carlini, N., Wagner, D.: Audio adversarial examples: Targeted attacks on
  speech-to-text. In: IEEE S\&P Workshops (2018)

\bibitem{chen2020boosting}
Chen, W., Zhang, Z., Hu, X., Wu, B.: Boosting decision-based black-box
  adversarial attacks with random sign flip. In: ECCV (2020)

\bibitem{chen2019deep}
Chen, Y., Lai, Z., Ding, Y., Lin, K., Wong, W.K.: Deep supervised hashing with
  anchor graph. In: CVPR (2019)

\bibitem{chen2018deep}
Chen, Z., Yuan, X., Lu, J., Tian, Q., Zhou, J.: Deep hashing via discrepancy
  minimization. In: CVPR (2018)

\bibitem{chua2009nus}
Chua, T.S., Tang, J., Hong, R., Li, H., Luo, Z., Zheng, Y.: Nus-wide: a
  real-world web image database from national university of singapore. In: ICMR
  (2009)

\bibitem{dong2018}
Dong, Y., Liao, F., Pang, T., Su, H., Zhu, J., Hu, X., Li, J.: Boosting
  adversarial attacks with momentum. In: CVPR (2018)

\bibitem{dong2019efficient}
Dong, Y., Su, H., Wu, B., Li, Z., Liu, W., Zhang, T., Zhu, J.: Efficient
  decision-based black-box adversarial attacks on face recognition. In: CVPR
  (2019)

\bibitem{duan2020adversarial}
Duan, R., Ma, X., Wang, Y., Bailey, J., Qin, A.K., Yang, Y.: Adversarial
  camouflage: Hiding physical-world attacks with natural styles. In: CVPR
  (2020)

\bibitem{eykholt2018robust}
Eykholt, K., Evtimov, I., Fernandes, E., Li, B., Rahmati, A., Xiao, C.,
  Prakash, A., Kohno, T., Song, D.: Robust physical-world attacks on deep
  learning visual classification. In: CVPR (2018)

\bibitem{sapf2020}
Fan, Y., Wu, B., Li, T., Zhang, Y., Li, M., Li, Z., Yang, Y.: Sparse
  adversarial attack via perturbation factorization. In: ECCV (2020)

\bibitem{feng2020universal}
Feng, Y., Chen, B., Dai, T., Xia, S.t.: Adversarial attack on deep product
  quantization network for image retrieval. In: AAAI (2020)

\bibitem{goodfellow2014}
Goodfellow, I.J., Shlens, J., Szegedy, C.: Explaining and harnessing
  adversarial examples. In: ICLR (2015)

\bibitem{gu2016supervised}
Gu, Y., Ma, C., Yang, J.: Supervised recurrent hashing for large scale video
  retrieval. In: ACM MM (2016)

\bibitem{hochreiter1997long}
Hochreiter, S., Schmidhuber, J.: Long short-term memory. Neural computation
  \textbf{9}(8),  1735--1780 (1997)

\bibitem{hu2018deep}
Hu, D., Nie, F., Li, X.: Deep binary reconstruction for cross-modal hashing.
  IEEE Transactions on Multimedia  \textbf{21}(4),  973--985 (2018)

\bibitem{inkawhich2019feature}
Inkawhich, N., Wen, W., Li, H.H., Chen, Y.: Feature space perturbations yield
  more transferable adversarial examples. In: CVPR (2019)

\bibitem{jhuang2013towards}
Jhuang, H., Gall, J., Zuffi, S., Schmid, C., Black, M.J.: Towards understanding
  action recognition. In: ICCV (2013)

\bibitem{krizhevsky2012imagenet}
Krizhevsky, A., Sutskever, I., Hinton, G.E.: Imagenet classification with deep
  convolutional neural networks. In: NeurIPS. pp. 1097--1105 (2012)

\bibitem{kurakin2016}
Kurakin, A., Goodfellow, I., Bengio, S.: Adversarial examples in the physical
  world. In: ICLR (2017)

\bibitem{lai2015simultaneous}
Lai, H., Pan, Y., Liu, Y., Yan, S.: Simultaneous feature learning and hash
  coding with deep neural networks. In: CVPR (2015)

\bibitem{li2018self}
Li, C., Deng, C., Li, N., Liu, W., Gao, X., Tao, D.: Self-supervised
  adversarial hashing networks for cross-modal retrieval. In: CVPR (2018)

\bibitem{li2019universal}
Li, J., Ji, R., Liu, H., Hong, X., Gao, Y., Tian, Q.: Universal perturbation
  attack against image retrieval. In: ICCV (2019)

\bibitem{li2012spectral}
Li, P., Wang, M., Cheng, J., Xu, C., Lu, H.: Spectral hashing with semantically
  consistent graph for image indexing. IEEE Transactions on Multimedia
  \textbf{15}(1),  141--152 (2012)

\bibitem{li2019neighborhood}
Li, S., Chen, Z., Lu, J., Li, X., Zhou, J.: Neighborhood preserving hashing for
  scalable video retrieval. In: ICCV. pp. 8212--8221 (2019)

\bibitem{liong2016deep}
Liong, V.E., Lu, J., Tan, Y.P., Zhou, J.: Deep video hashing. IEEE Transactions
  on Multimedia  \textbf{19}(6),  1209--1219 (2016)

\bibitem{liu2016deep}
Liu, H., Wang, R., Shan, S., Chen, X.: Deep supervised hashing for fast image
  retrieval. In: CVPR (2016)

\bibitem{liu2012supervised}
Liu, W., Wang, J., Ji, R., Jiang, Y.G., Chang, S.F.: Supervised hashing with
  kernels. In: CVPR (2012)

\bibitem{ma2019understanding}
Ma, X., Niu, Y., Gu, L., Wang, Y., Zhao, Y., Bailey, J., Lu, F.: Understanding
  adversarial attacks on deep learning based medical image analysis systems.
  Pattern Recognition  (2019)

\bibitem{moosavi2016}
Moosavi-Dezfooli, S.M., Fawzi, A., Frossard, P.: Deepfool: a simple and
  accurate method to fool deep neural networks. In: CVPR (2016)

\bibitem{paszke2019pytorch}
Paszke, A., Gross, S., Massa, F., Lerer, A., Bradbury, J., Chanan, G., Killeen,
  T., Lin, Z., Gimelshein, N., Antiga, L., et~al.: Pytorch: An imperative
  style, high-performance deep learning library. In: NeurIPS (2019)

\bibitem{qin2019imperceptible}
Qin, Y., Carlini, N., Cottrell, G., Goodfellow, I., Raffel, C.: Imperceptible,
  robust, and targeted adversarial examples for automatic speech recognition.
  In: ICML (2019)

\bibitem{russakovsky2015imagenet}
Russakovsky, O., Deng, J., Su, H., Krause, J., Satheesh, S., Ma, S., Huang, Z.,
  Karpathy, A., Khosla, A., Bernstein, M., et~al.: Imagenet large scale visual
  recognition challenge. International journal of computer vision
  \textbf{115}(3),  211--252 (2015)

\bibitem{shen2015supervised}
Shen, F., Shen, C., Liu, W., Tao~Shen, H.: Supervised discrete hashing. In:
  CVPR (2015)

\bibitem{shen2018unsupervised}
Shen, F., Xu, Y., Liu, L., Yang, Y., Huang, Z., Shen, H.T.: Unsupervised deep
  hashing with similarity-adaptive and discrete optimization. IEEE transactions
  on pattern analysis and machine intelligence  \textbf{40}(12),  3034--3044
  (2018)

\bibitem{simonyan2014very}
Simonyan, K., Zisserman, A.: Very deep convolutional networks for large-scale
  image recognition. In: ICLR (2015)

\bibitem{song2018self}
Song, J., Zhang, H., Li, X., Gao, L., Wang, M., Hong, R.: Self-supervised video
  hashing with hierarchical binary auto-encoder. IEEE Transactions on Image
  Processing  \textbf{27}(7),  3210--3221 (2018)

\bibitem{soomro2012ucf101}
Soomro, K., Zamir, A.R., Shah, M.: Ucf101: A dataset of 101 human actions
  classes from videos in the wild. arXiv preprint arXiv:1212.0402  (2012)

\bibitem{szegedy2013}
Szegedy, C., Zaremba, W., Sutskever, I., Bruna, J., Erhan, D., Goodfellow, I.,
  Fergus, R.: Intriguing properties of neural networks. In: ICLR (2014)

\bibitem{tolias2019targeted}
Tolias, G., Radenovic, F., Chum, O.: Targeted mismatch adversarial attack:
  Query with a flower to retrieve the tower. In: ICCV (2019)

\bibitem{tramer2017ensemble}
Tram{\`e}r, F., Kurakin, A., Papernot, N., Goodfellow, I., Boneh, D., McDaniel,
  P.: Ensemble adversarial training: Attacks and defenses. In: ICLR (2018)

\bibitem{wang2017survey}
Wang, J., Zhang, T., Sebe, N., Shen, H.T., et~al.: A survey on learning to
  hash. IEEE transactions on pattern analysis and machine intelligence
  \textbf{40}(4),  769--790 (2017)

\bibitem{wang2015learning}
Wang, J., Liu, W., Kumar, S., Chang, S.F.: Learning to hash for indexing big
  data—a survey. Proceedings of the IEEE  \textbf{104}(1),  34--57 (2015)

\bibitem{wu2017deep}
Wu, D., Lin, Z., Li, B., Ye, M., Wang, W.: Deep supervised hashing for
  multi-label and large-scale image retrieval. In: ICMR (2017)

\bibitem{wu2020skip}
Wu, D., Wang, Y., Xia, S.T., Bailey, J., Ma, X.: Skip connections matter: On
  the transferability of adversarial examples generated with resnets. In: ICLR
  (2020)

\bibitem{wu2018unsupervised}
Wu, G., Han, J., Guo, Y., Liu, L., Ding, G., Ni, Q., Shao, L.: Unsupervised
  deep video hashing via balanced code for large-scale video retrieval. IEEE
  Transactions on Image Processing  \textbf{28}(4),  1993--2007 (2018)

\bibitem{xia2014supervised}
Xia, R., Pan, Y., Lai, H., Liu, C., Yan, S.: Supervised hashing for image
  retrieval via image representation learning. In: AAAI (2014)

\bibitem{xu2019exact}
Xu, Y., Wu, B., Shen, F., Fan, Y., Zhang, Y., Shen, H.T., Liu, W.: Exact
  adversarial attack to image captioning via structured output learning with
  latent variables. In: CVPR (2019)

\bibitem{xu2019}
Xu, Y., Wu, B., Shen, F., Fan, Y., Zhang, Y., Shen, H.T., Liu, W.: Exact
  adversarial attack to image captioning via structured output learning with
  latent variables. In: CVPR (2019)

\bibitem{yan2017semi}
Yan, X., Zhang, L., Li, W.J.: Semi-supervised deep hashing with a bipartite
  graph. In: IJCAI (2017)

\bibitem{yang2019distillhash}
Yang, E., Liu, T., Deng, C., Liu, W., Tao, D.: Distillhash: Unsupervised deep
  hashing by distilling data pairs. In: CVPR (2019)

\bibitem{yang2018adversarial}
Yang, E., Liu, T., Deng, C., Tao, D.: Adversarial examples for hamming space
  search. IEEE transactions on cybernetics  \textbf{50}(4),  1473--1484 (2018)

\bibitem{yao2019}
Yao, Z., Gholami, A., Xu, P., Keutzer, K., Mahoney, M.W.: Trust region based
  adversarial attack on neural networks. In: CVPR (2019)

\bibitem{zhang2004solving}
Zhang, T.: Solving large scale linear prediction problems using stochastic
  gradient descent algorithms. In: ICML (2004)

\bibitem{zhao2015deep}
Zhao, F., Huang, Y., Wang, L., Tan, T.: Deep semantic ranking based hashing for
  multi-label image retrieval. In: CVPR (2015)

\bibitem{zhu2016deep}
Zhu, H., Long, M., Wang, J., Cao, Y.: Deep hashing network for efficient
  similarity retrieval. In: AAAI (2016)

\bibitem{zuva2012evaluation}
Zuva, K., Zuva, T.: Evaluation of information retrieval systems. International
  journal of computer science \& information technology  \textbf{4}(3), ~35
  (2012)

\end{thebibliography}

\newpage

\begin{center}
    \begin{Large}
        \textbf{Appendix}
    \end{Large}
\end{center}

\begin{subappendices}
\setcounter{lemma}{0}
\setcounter{theorem}{0}
\renewcommand{\thesection}{\Alph{section}}

\section{Proof of Theorem 1}
\begin{theorem} \label{thm1}
 Anchor code ${\bm{h}_a}$ calculated by Algorithm 1 is the binary code achieving the minimal sum of Hamming distances with respect to ${\bm{h}_i}$, $i = 1, \dots, n_t$, i.e., 
 \begin{equation}
 {\bm{h}_a}=\arg \min_{\bm{h}\in\{+1,-1\}^K}\sum_{i=1}^{n_t}d_H({\bm{h}}, {\bm{h}_i}).
 \end{equation}
\end{theorem}

\begin{proof}
We only need to prove that for any ${{\bm{h}}\in\{+1,-1\}}^K$ and ${\bm{ h}}\neq {\bm{h}_a}$, the following inequality holds.
\begin{equation}
\sum_{i}^{n_t}d_H({\bm{h}_a}, {\bm{h}_i})\leq \sum_{i}^{n_t}d_H({\bm{h}}, {\bm{h}_i}).
\end{equation}
Denote $\mathcal{D}=\{j_1, j_2,\dots, j_{K_0}\}$, $1\leq K_0\leq K$, as the index set where  ${\bm{h}}$ and  ${\bm{h}_a}$ differ. Then we have 
\begin{eqnarray}
&\quad&\sum_{i}^{n_t}d_H({\bm{h}_a}, {\bm{ h}_i})\nonumber\\
&=&\sum_{j\in \mathcal{D}}d_H({\bm{h}_a^j}, {\bm{h}_i^j})+\sum_{j\in \{1, 2, \dots, K\}\setminus\mathcal{D}}d_H({\bm{h}_a^j}, {\bm{h}_i^j})\\
&=&\sum_{j\in \mathcal{D}}\big(n_t-\sum_{i}^{n_t}\mathbb{I}({\bm{h}_a^j}={\bm{h}_i^j})\big)+\sum_{j\in \{1, 2, \dots, K\}\setminus\mathcal{D}}\big(n_t-\sum_{i}^{n_t}\mathbb{I}({\bm{ h}_a^j}={\bm{ h}_i^j})\big)\\
&\overset{(a)}{\leq}& \sum_{j\in \mathcal{D}}\big(n_t-\sum_{i}^{n_t}\mathbb{I}({\bm{h}^j}={\bm{h}_i^j})\big)+\sum_{j\in \{1, 2, \dots, K\}\setminus\mathcal{D}}\big(n_t-\sum_{i}^{n_t}\mathbb{I}({\bm{h}^j}={\bm{h}_i^j})\big)\\
&=&\sum_{i}^{n_t}d_H({\bm{h}}, {\bm{h}_i}),
\end{eqnarray}
where $(a)$ holds since anchor code $\bm{h}_a$ is obtained through a voting process (as shown in Algorithm 1 in the main manuscript), $i.e.$,  $ \forall j\in \mathcal{D}$, 
\begin{equation}
    \sum_{i=1}^{n_t} \mathbb{I}({\bm{h}_a^j}={\bm{h}_i^j})\geq \sum_{i=1}^{n_t} \mathbb{I}({\bm{h}^j}={\bm{h}_i^j}).
\end{equation}
$\hfill\blacksquare$ 

\end{proof}

\section{Threat Models}

All experiments are implemented on the PyTorch framework \cite{paszke2019pytorch}. 
The detailed training settings are shown as follows.

\vspace{0.5em}
\noindent \textbf{Image Hashing.}
We adopt VGG-11 \cite{simonyan2014very} as the backbone network pre-trained on ImageNet to extract features, then replace the last fully-connected layer of softmax classifier with the hashing layer. We fine-tune the base model and train the hash layer from scratch through the pairwise loss function in \cite{yang2018adversarial}. We employ stochastic gradient descent (SGD) \cite{zhang2004solving} with momentum 0.9 as the optimizer. The weight decay parameter is set to 0.0005. The learning rate is fixed at 0.01 and the batch size is 24.

\vspace{0.5em}
\noindent \textbf{Video Hashing.} We extract frame features using AlexNet \cite{krizhevsky2012imagenet} pretrained on the ImageNet dataset. Then we employ the objective function in \cite{liong2016deep} to train LSTM \cite{hochreiter1997long} with the hash layer from scratch. The parameter in the objective function to balance discriminative loss and quantization loss is set to 0.0001. SGD is used to optimize model parameters, with the momentum 0.9 and the fixed learning rate 0.05. The weight decay parameter is set to 0.0001. The batch size is set to 100 and the maximum length of input videos is 40. Due to different video sizes for two video datasets, we adopt different strategies to sample video frames. For the JHMDB dataset, we select all frames of a video whose length is smaller than 40  and top-40 frames otherwise. For the UCF-101 dataset, we select frames with equal stride (set to 3) for each video.

\section{Datasets Description}

Four retrieval benchmark datasets are adopted in our experiments. The first two datasets are used for image retrieval, while the last two are used for video retrieval. These datasets are described in details as follows.

\begin{itemize}
    \item \emph{ImageNet} \cite{russakovsky2015imagenet} consists of 1.2M training samples and 50,000 testing samples with 1000 classes. We follow \cite{cao2017hashnet} to build a subset containing 130K images with 100 classes. We use images from the training set as the database, and images from the testing set as the queries. We sample 100 images per class from the database for training the deep hashing model. 
    \item \emph{NUS-WIDE} \cite{chua2009nus} dataset contains 269,648 images from 81 classes. We only select a subset of images with the 20 most frequent labels. We randomly sample 5000 images as the query set and take the remaining images as the database, following \cite{zhu2016deep}. Besides, we randomly sample 10,000 images from the database to train the hashing model.  
    \item \emph{JHMDB} \cite{jhuang2013towards} consists of 928 videos in 21 categories. We randomly choose 10 videos per category as queries, 10 videos per category as training samples, and the rest as retrieval database. 
    \item \emph{UCF-101} \cite{soomro2012ucf101} is an action recognition dataset, which contains 13,320 videos categorized into 101 classes. We use 30 videos per category for training, 30 videos per category for querying and the remaining 7,260 videos as the database.
\end{itemize}




\begin{figure}[!ht]
\vspace{1em}
\centering
\subfigure[JHMDB]{
\includegraphics[width=0.99\textwidth]{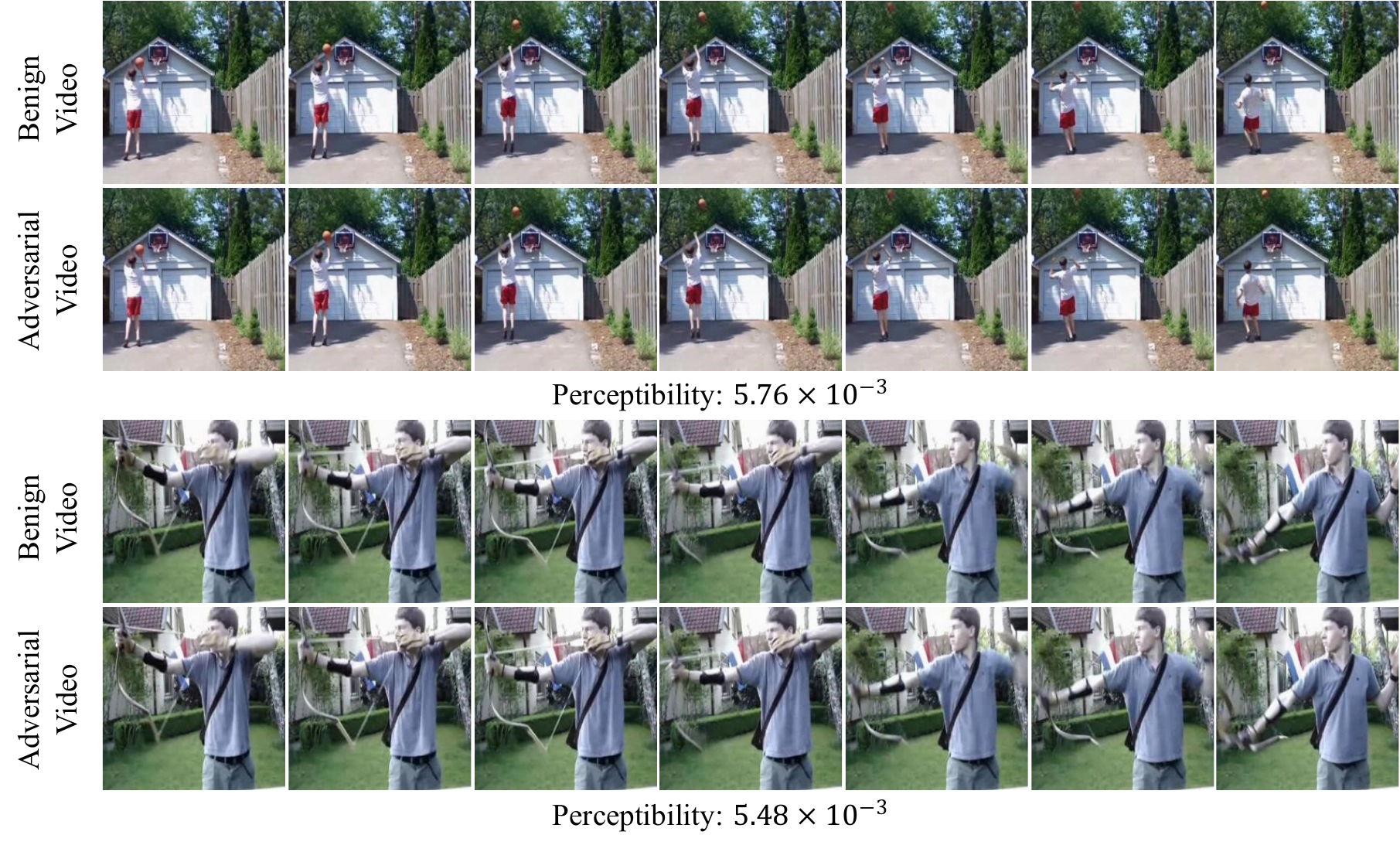}}
\subfigure[UCF-101]{
\includegraphics[width=0.99\textwidth]{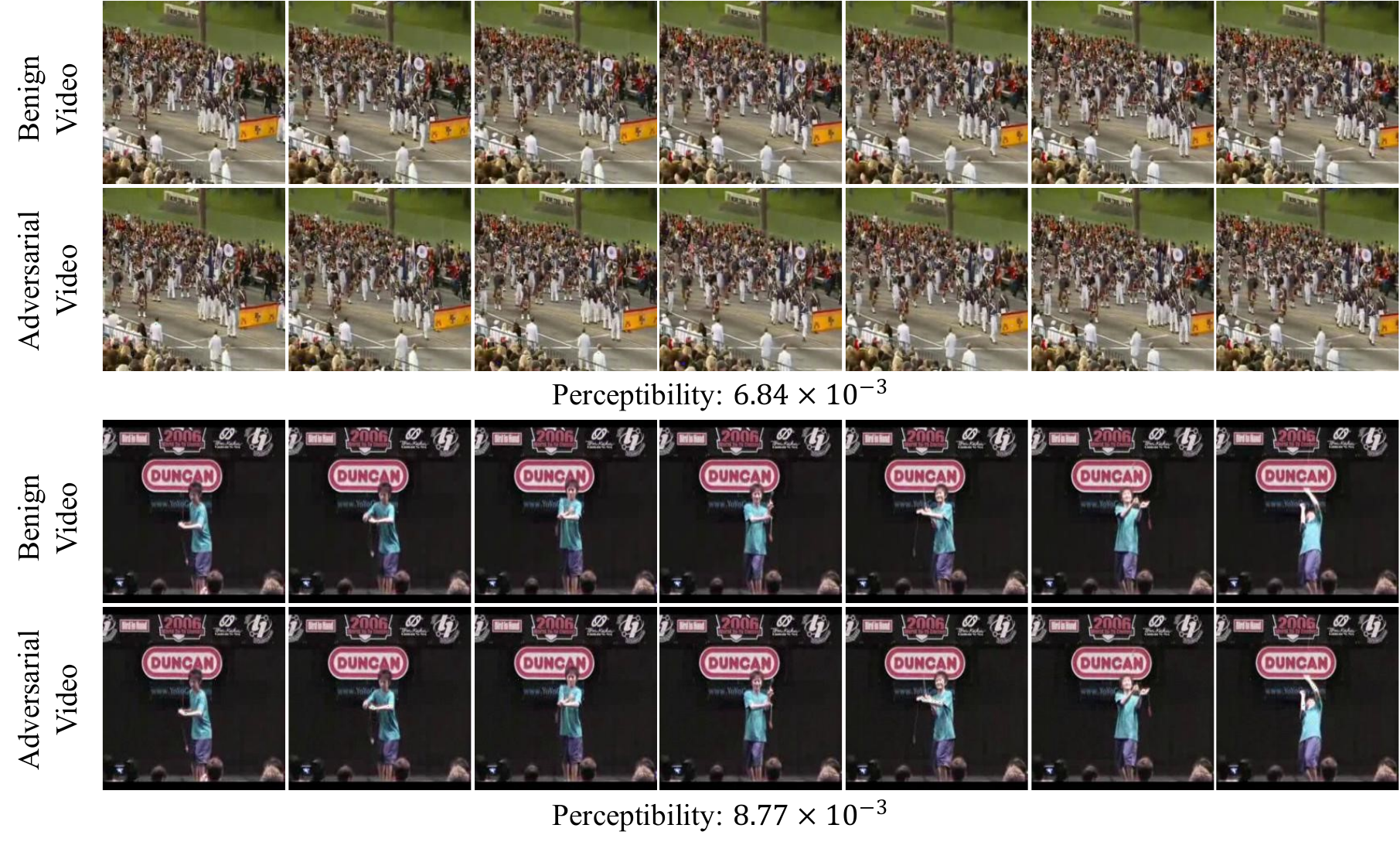}}
\caption{Visualization examples of generated adversarial examples in video hashing.}
\label{fig:adv_video_exp}
\vspace{3em}
\end{figure}

\begin{figure}[ht]
	\vspace{1em}
	\centering
	\includegraphics[width=0.99\textwidth]{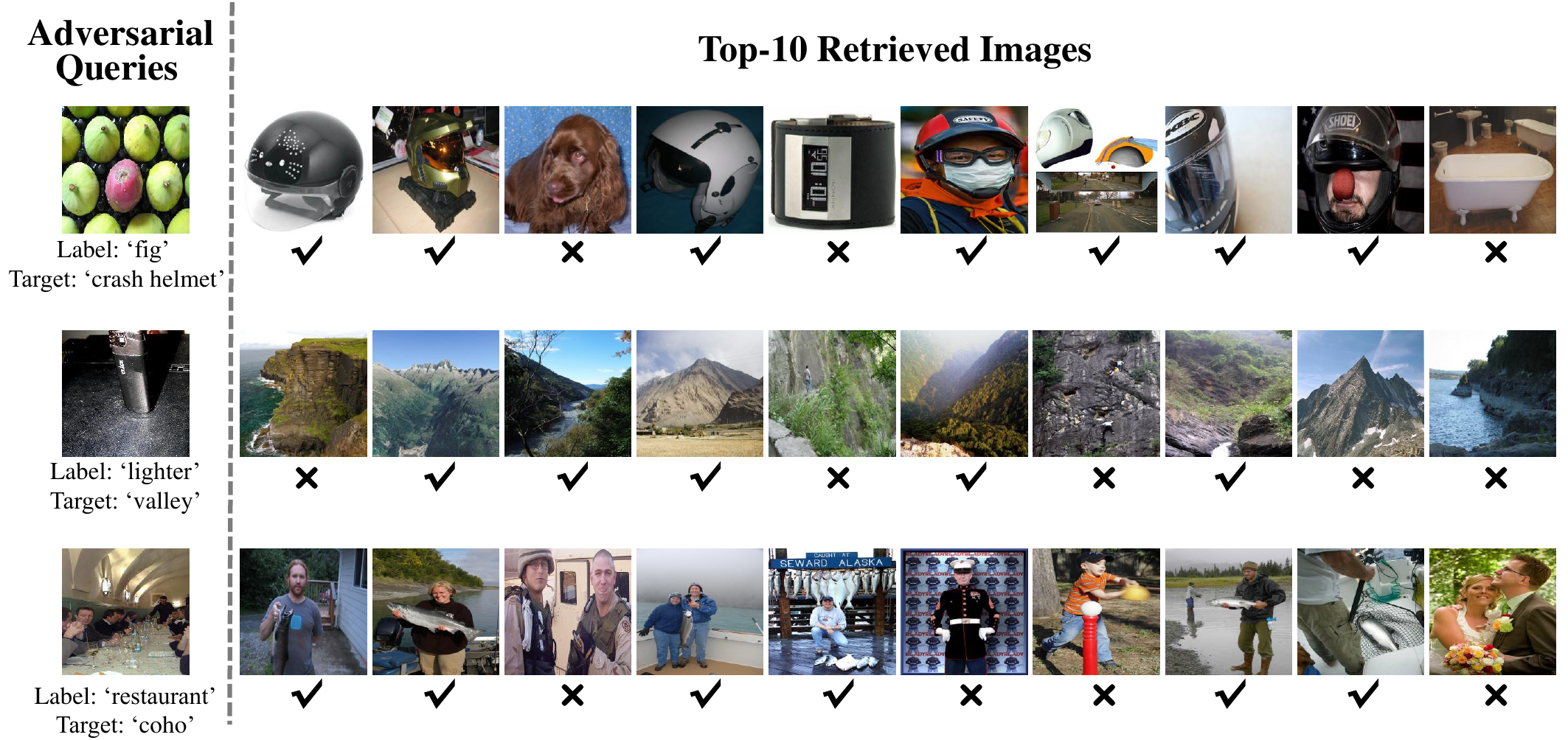}
	
	\caption{Examples of image retrieval with adversarial query on ImageNet. All target labels are randomly selected from the out-of-sample class labels. Retrieved objects with top-10 similarity are shown on the right. The tick and cross indicate whether the retrieved object is consistent with the target label.}
	\label{fig:open-set}
\end{figure}

\section{Visualization}
In this section, we provide some visual examples of DHTA in video hashing and open-set scenario.

\vspace{0.5em}
\noindent \textbf{Video Hashing.} 
Some examples of generated adversarial videos and their correspondingly benign videos are shown in Figure \ref{fig:adv_video_exp}. Specifically, we present frames with indexes $\in \{3,6,9,12,15,18,21\}$ for each video due to the limitation of the space. Similar to the image scenario, these visual results show that the adversarial queries are very similar to their original versions. In other words, the generated adversarial objects of our proposed DHTA are human-imperceptible.


\vspace{0.5em}
\noindent \textbf{Open-set Targeted Attack.}
We demonstrate some generated adversarial examples and their corresponding retrieved images under an open-set scenario in Figure \ref{fig:open-set}. Even if this setting is tougher, there still exist some images with targeted label in the top-10 retrieved images. This result reveals that our proposed DHTA can successfully fool deep hashing model to return objects from out-of-sample class.

\section{Further Discussion}

\noindent \textbf{Attack Towards the Advanced Model.}
To verify the effectiveness of the proposed method in attacking the advanced deep hashing based retrieval model, we conduct experiments against HashNet \cite{cao2017hashnet} on the NUS-WIDE dataset. Evaluation settings are the same as those used in Section 4.2. As shown in Table \ref{tab:hashnet}, DHTA can successfully attack HashNet with the high t-MAP. Compared with P2P, the t-MAP improvement of DHTA is over 7\% in all cases. Moreover, the t-MAP value of DHTA is significantly higher than the MAP value of the `Original'. These results also verify the high effectiveness of our DHTA method.

\begin{table}[t]
	\centering
    \caption{t-MAP (\%) of targeted attack methods and MAP (\%) of query with benign objects (‘Original’) with various code lengths on NUS-WIDE dataset.}
	\setlength{\tabcolsep}{1.5mm}{
        \begin{tabular}{cccccc}
        \hline
        Method   & Metric & 16bits & 32bits & 48bits & 64bits \\ \hline
        Original & t-MAP  & 43.28  & 41.88  & 43.44  & 43.55  \\
        Noise    & t-MAP  & 42.67  & 39.86  & 42.50  & 41.56  \\
        P2P      & t-MAP  & 78.10  & 81.79  & 81.74  & 82.68  \\
        DHTA     & t-MAP  & \bf 86.95  & \bf 89.02  & \bf 89.68  & \bf 90.49  \\ \hline
        Original & MAP    & 79.95  & 81.88  & 83.11  & 84.96  \\ \hline
        \end{tabular}
    }
	\label{tab:hashnet}
\end{table}

\begin{table}[t]
	\centering
    \caption{t-MAP (\%) of DHTA with different $\epsilon$ under 32-bits code length on ImageNet and JHMDB dataset.}
	\setlength{\tabcolsep}{1.5mm}{
        \begin{tabular}{ccccccc}
        \hline
        $\epsilon$ & 0.01  & 0.02  & 0.03  & 0.04  & 0.05   \\ \hline
        ImageNet & 25.39 & 61.32 & 75.06 & 79.01 & 79.17  \\ \hline
        JHMDB    & 51.08 & 60.20  & 61.23 & 61.75 & 62.76  \\ \hline
        \end{tabular}
    }
	\label{tab:eps}
\end{table}

\begin{table}[!t]
	\centering
    \caption{t-MAP (\%) of different attack methods on ImageNet and JHMDB dataset.}
	\setlength{\tabcolsep}{1mm}{
        \begin{tabular}{ccccclcccc}
        \hline
        \multirow{2}{*}{Method} & \multicolumn{4}{c}{ImageNet}      &  & \multicolumn{4}{c}{JHMDB}         \\ \cline{2-5} \cline{7-10} 
                                & 16bits & 32bits & 48bits & 64bits &  & 16bits & 32bits & 48bits & 64bits \\ \hline
        Feature-based Attack    & 23.80  & 28.47  & 34.24  & 33.05  &  & 30.47  & 44.95  & 42.48  & 57.54  \\
        DHTA                    & 63.68  & 77.76  & 82.31  & 82.10   &  & 56.47  & 62.04  & 63.02  & 66.06  \\ \hline
        \end{tabular}
    }
	\label{tab:fea}
\end{table}

\vspace{0.6em}
\noindent \textbf{Effect of the Maximum Perturbation Strength.}
To analyze the effect of the maximum perturbation strength ($i.e., \epsilon$), we examine the t-MAP of DHTA under different values of $\epsilon \in \{0.01, 0.02, 0.03, 0.04, 0.05\}$. Table \ref{tab:eps} presents the t-MAP of DHTA under 32-bits code length on ImageNet and JHMDB datasets. It can be seen that the attack performance (t-MAP) improves as the increase of $\epsilon$, which demonstrates the trade-off between the perceptibility of adversarial perturbations and attack performance.

\vspace{0.6em}
\noindent \textbf{Comparison with Feature-based Attack.}
Although many adversarial attacks in the image recognition are proposed, most of them cannot be directly adopted due to the property of the retrieval. The feature-based attack \cite{inkawhich2019feature}, as an exception, can be naturally extended to attack deep hashing models. The main idea of feature-based attack is to make the adversarial image closed to an image of the target class in the feature space. We choose the intermediate feature before the hashing layer to perform the feature-based attack. The comparison between feature-based attack and DHTA on the ImageNet and JHMDB datasets is shown in Tabel \ref{tab:fea}. 
There exists a large gap (around 40\% for ImageNet and 20\% for JHMDB) between the feature-based attack and DHTA. Such a superior result of DHTA reveals that, manipulating in Hamming space is more effective than feature space for attacking deep hashing model.

\end{subappendices}

\end{document}